\documentclass[journal, twocol]{IEEEtran}
\usepackage{amsthm}
\usepackage{xcolor}
\usepackage{amsmath, amsfonts}
\usepackage{mathtools}
\usepackage{multicol}
\usepackage{stmaryrd}
\usepackage{subcaption}
\usepackage{graphicx}
\usepackage{cuted}
\usepackage{stfloats}
\usepackage{tikz}
\usetikzlibrary{arrows.meta,calc,patterns,positioning,shapes.geometric,arrows,decorations.text}

\makeatletter
\let\save@mathaccent\mathaccent
\newcommand*\if@single[3]{%
  \setbox0\hbox{${\mathaccent"0362{#1}}^H$}%
  \setbox2\hbox{${\mathaccent"0362{\kern0pt#1}}^H$}%
  \ifdim\ht0=\ht2 #3\else #2\fi
  }
\newcommand*\rel@kern[1]{\kern#1\dimexpr\macc@kerna}
\newcommand*\wideaccent[2]{\@ifnextchar^{{\wide@accent{#1}{#2}{0}}}{\wide@accent{#1}{#2}{1}}}
\newcommand*\wide@accent[3]{\if@single{#2}{\wide@accent@{#1}{#2}{#3}{1}}{\wide@accent@{#1}{#2}{#3}{2}}}
\newcommand*\wide@accent@[4]{%
  \begingroup
  \def\mathaccent##1##2{%
    \let\mathaccent\save@mathaccent
    \if#42 \let\macc@nucleus\first@char \fi
    \setbox\z@\hbox{$\macc@style{\macc@nucleus}_{}$}%
    \setbox\tw@\hbox{$\macc@style{\macc@nucleus}{}_{}$}%
    \dimen@\wd\tw@
    \advance\dimen@-\wd\z@
    \divide\dimen@ 3
    \@tempdima\wd\tw@
    \advance\@tempdima-\scriptspace
    \divide\@tempdima 10
    \advance\dimen@-\@tempdima
    \ifdim\dimen@>\z@ \dimen@0pt\fi
    \rel@kern{0.6}\kern-\dimen@
    \if#41
      #1{\rel@kern{-0.6}\kern\dimen@\macc@nucleus\rel@kern{0.4}\kern\dimen@}%
      \advance\dimen@0.4\dimexpr\macc@kerna
      \let\final@kern#3%
      \ifdim\dimen@<\z@ \let\final@kern1\fi
      \if\final@kern1 \kern-\dimen@\fi
    \else
      #1{\rel@kern{-0.6}\kern\dimen@#2}%
    \fi
  }%
  \macc@depth\@ne
  \let\math@bgroup\@empty \let\math@egroup\macc@set@skewchar
  \mathsurround\z@ \frozen@everymath{\mathgroup\macc@group\relax}%
  \macc@set@skewchar\relax
  \let\mathaccentV\macc@nested@a
  \if#41
    \macc@nested@a\relax111{#2}%
  \else
    \def\gobble@till@marker##1\endmarker{}%
    \futurelet\first@char\gobble@till@marker#2\endmarker
    \ifcat\noexpand\first@char A\else
      \def\first@char{}%
    \fi
    \macc@nested@a\relax111{\first@char}%
  \fi
  \endgroup
}
\makeatother

\newcommand{\oa}{{\overline{a}}}

\newcommand{\bff}{{\mathbf{f}}}

\newcommand{\oh}{{\overline{h}}}

\newcommand{\bfi}{{\mathbf{i}}}

\newcommand{\tl}{{\widetilde{l}}}

\newcommand{\tn}{{\widetilde{n}}}

\newcommand{\op}{{\overline{p}}}

\newcommand{\tp}{{\widetilde{p}}}
\newcommand{\hp}{{\hat{p}}}

\newcommand{\oy}{{\overline{y}}}

\newcommand{\mA}{{\mathcal{A}}}
\newcommand{\omA}{{\overline{\mathcal{A}}}}

\newcommand{\mB}{{\mathcal{B}}}

\newcommand{\bfC}{{\mathbf{C}}}

\newcommand{\bD}{{\mathbb{D}}}

\newcommand{\bfE}{{\mathbf{E}}}

\newcommand{\bF}{{\mathbb{F}}}

\newcommand{\oH}{{\overline{H}}}

\newcommand{\bfH}{{\mathbf{H}}}
\newcommand{\tH}{{\widetilde{H}}}

\newcommand{\bfI}{{\mathbf{I}}}

\newcommand{\bN}{{\mathbb{N}}}
\newcommand{\bfN}{{\mathbf{N}}}

\newcommand{\bR}{{\mathbb{R}}}
\newcommand{\bfR}{{\mathbf{R}}}

\newcommand{\mR}{{\mathcal{R}}}

\newcommand{\bfS}{{\mathbf{S}}}

\newcommand{\mS}{{\mathcal{S}}}

\newcommand{\oV}{{\overline{V}}}

\newcommand{\bfV}{{\mathbf{V}}}
\newcommand{\tV}{{\widetilde{V}}}

\newcommand{\tW}{{\widetilde{W}}}

\newcommand{\mX}{{\mathcal{X}}}

\newcommand{\mY}{{\mathcal{Y}}}
\newcommand{\omY}{{\overline{\mathcal{Y}}}}

\newcommand{\tmY}{{\widetilde{\mY}}}

\newcommand{\bZ}{{\mathbb{Z}}}

\newcommand{\mZ}{{\mathcal{Z}}}

\newcommand{\abs}[1]{\left|{#1}\right|}
\newcommand{\set}[1]{{\left\{{#1}\right\}}}
\newcommand{\br}[1]{{\left({#1}\right)}}
\newcommand{\brsq}[1]{{\left[{#1}\right]}}

\newcommand{\brabs}[1]{{\left|{#1}\right|}}

\newcommand{\prob}[1]{\set{#1}}

\newcommand{\Probd}[2]{{\Pr_{#1}\prob{#2}}}

\newcommand{\ceil}[1]{\lceil{{#1}\rceil}}
\newcommand{\floor}[1]{\lfloor{{#1}\rfloor}}
\newcommand{\defeq}{{\overset{\triangle}{=}}}

\newcommand{\ve}{\varepsilon}
\newcommand{\vp}{\varphi}
\newcommand{\tosq}{\rightsquigarrow}
\newcommand{\toshort}{\shortrightarrow}
\newcommand{\bg}{\mathbf{g}}

\newcommand{\bi}{\mathbf{i}}

\newcommand{\obeta}{\overline{\beta}}

\newtheorem{theorem}{Theorem}
\newtheorem{proposition}{Proposition}

\DeclareMathOperator{\const}{const}

\newcommand{\pW}[2]{[{#1}{#2}]}
\newcommand{\pWd}[2]{[{#1}{#2}]}
\newcommand{\pWy}[2]{[{#1}{#2}]}
\newcommand{\pWyd}[2]{[{#1}{#2}]}
\newcommand{\vc}[1]{\mathbf{#1}}

\graphicspath{{eps/}}

\title{Channels with Markov Synchronization Errors: Information Stability and Capacity Bounds}
\author{Ruslan Morozov and Tolga M. Duman\\Bilkent University, Ankara, Turkey\thanks{This work was funded by the European Union through the ERC Advanced Grant 101054904: TRANCIDS. Views and opinions expressed are, however, those of the authors only and do not necessarily reflect those of the European Union or the European Research Council Executive Agency. Neither the European Union nor the granting authority can be held responsible for them.}
\thanks{This paper was presented in part at the 2024 IEEE International Symposium on Information Theory (ISIT) \cite{Morozov2024ISIT}.}\\Email: ruslan.morozov@bilkent.edu.tr, duman@ee.bilkent.edu.tr} 
\date{}

\begin{document}

\maketitle 
\begin{abstract}
Particularly motivated by DNA storage channels, we consider channels with synchronization errors modeled as insertions and deletions, along with substitutions. We focus on the case where the synchronization error process has memory and investigate the information stability of these channels, hence the existence of their Shannon capacity. We assume that the synchronization errors are governed by a stationary and ergodic finite state Markov chain and prove that such a channel is information-stable, which implies the existence of a coding scheme that achieves the limit of mutual information. This result implies the existence of the Shannon capacity for a wide range of channels with synchronization errors, with different applications, including DNA storage. We also provide specific examples of deletion channels with Markov memory and numerically evaluate their capacity bounds, thereby allowing us to quantify the capacity difference between memoryless deletion channels and those with memory with the same deletion probability and reveal that having memory increases the channel capacity.
\end{abstract}

\section{Introduction}
\label{s:intr}

Channels with synchronization errors are observed in different communication systems. For instance, for a wireless communication system at low signal-to-noise ratios, loss of synchronization results in cycle slips, causing mistimed sampling of the matched filter output, resulting in possible deletions, insertions, and substitutions. Hence, we refer to channels with synchronization errors as insertion/deletion/substitution (IDS) channels. As another example, misalignment of the read-head and write-head may cause deletion or insertion errors for bit-patterned media recording systems \cite{hu07bit}. The insertion/deletion errors are in addition to other possible errors due to noise, modeled as an additive white Gaussian noise process or possibly asymmetric substitution errors. The emerging area of DNA storage represents another scenario in which IDS channels are observed in a practical setting \cite{shomorony21dna}.

There is extensive literature studying different aspects of transmission over channels with insertions/deletions. For instance, there are a number of number theoretic and algebraic code designs (suited for relatively short transmissions) (see, for example, \cite{Sima2021,Song2022,Song2023} for some recent examples) and applications of concatenated coding solutions with inner marker or watermark codes \cite{Ratzer2000,Davey2001,Wang2011}. On a different front, these channels have also been studied from an information-theoretic point of view. 
In particular, Dobrushin \cite{dobrushin67shannons} studied the case, when the synchronization errors are memoryless, i.e., while for transmission of each channel input, the number of output symbols is not the same (hence the receiver cannot immediately identify which output symbols correspond to which input symbol), these errors are independent and identically distributed for all channel uses.
He proved that such synchronization error channel is information stable, and hence the channel capacity exists. Specifically, the capacity expression is given by $C = \mathop{\lim}_{n \rightarrow \infty} \mathop{\max}_{p(\bf{X})} \frac{I(\bf{X};\bf{Y})}{n}$ where $\bf{X}$ is the length-$n$ channel input, and $\bf{Y}$ is the channel output, $I(\cdot;\cdot)$ is the mutual information between the input and the output sequences, and the maximization is over all stationary and ergodic processes.
In \cite{ahlswede71channels}, the coding theorems and the strong converse results are presented for the synchronization error channel and for its generalization which has a feedback, i.e., after each transmitted input symbol the receiver sends the received sequence to the transmitter and transmitter can choose the subsequent symbol to transmit based on this information.

While it is known that the Shannon capacity exists for memoryless synchronization error channels, it has been formidable to evaluate the above expression and compute the capacity, even for some seemingly simple models. 
For instance, the capacity of the independent and identically distributed (i.i.d.) deletion channel is still unknown.
There are, however, upper and lower bounds on the capacity of different synchronization error channels developed over the years.
For instance, Refs. \cite{drinea06lower, Mitzenmacher2006, Driena2007} develop lower bounds on the capacity of i.i.d. deletion channels and channels with i.i.d. deletions and duplications. 
On the other hand, different works including \cite{Diggavi2007,fertonani10novel,rahmati15upper} develop upper bounds on the capacity of i.i.d. deletion channels as well as channels with deletions/insertions/substitutions. 
Asymptotic expansion of the i.i.d. deletion channel capacity is determined in \cite{kanoria13optimal,kalai10tight}. 
Excellent surveys of the information-theoretic developments on channels with synchronization errors can be found in \cite{mitzenmacher09survey,cheraghchi21overview}.

Based on experimental results, it has recently been demonstrated that the insertion/deletion processes in DNA storage channels exhibit memory \cite{hamoum23channel,deng19optimized}. 
In other words, the assumption of i.i.d. insertions and/or deletions falls short of some realistic channel models. 
Motivated by this fact, in this paper, we consider IDS channels with memory; specifically, we study the case where the synchronization errors are governed by an underlying Markov chain, referred to as Markov-IDS channels. 
That is, each state of the Markov chain is an IDS channel, called a component channel, and the overall output of the Markov-IDS channel is obtained by concatenating the individual component channel outputs. 
As the insertion/deletion processes have memory, Dobrushin's result is not directly applicable. 
Hence, we need a new coding theorem to prove that the Shannon capacity of such channels exists, which is the main problem addressed in the paper. 
To prove our results, we extract different steps in the proof of the capacity theorem for synchronization error channels as independent propositions to allow for their use in the future in a streamlined manner. 
Furthermore, we provide several new results, stated as general propositions, on the perturbation of an optimal input distribution for a channel and the order of growth of the variance of mutual information density in connection to the smallest non-zero channel transition probability.

Note that the Markov-IDS channel is a generalization of both IDS and Finite State Markov Channels (FSMC): it is more general than IDS since statistics of the synchronization errors can change over time; it is more general than FSMC since it concatenates the outputs of each component IDS channel (which can be translated to non-concatenated version via inserting a special symbol at the end of each component channel's output). 
We highlight that the most closely related work to our paper is \cite{mao18models}, which proves the capacity theorem for IDS channels with a state, which is defined by a few ($k$) recent input bits.
We call such channels the $k$-IDS model.
In the $k$-IDS model, by using filler bits at the input, the total channel can be translated to multiple independent uses of the same channel.
The Markov-IDS model, however, can only be translated to a product of similar (but not the same) independent channels.
This made the analysis of the variance of mutual information density necessary.
We also highlight \cite{con2025channels}, which generalizes the coding theorem from $k$-IDS to a class of \emph{admissable} channels.
The symbolwise mutual information\footnote{The symbolwise mutual information is defined as mutual information between the input symbol of the channel and the log-likelihood ratio (LLRs) of this symbol, produced by a symbolwise MAP decoder.} is estimated for some cases of $k$-IDS channels in \cite{maarouf23achievable}.

As a more general model, consider the case when the transitions of the Markov chain also depend on a constant number of last input bits as well.
Such a channel model would include both Markov-IDS and ISI-IDS models as special cases. The techniques developed in our paper may (potentially) address this general channel model as well. This is because we already use an infinite number of filler bits to split the Markov-IDS channel into independent copies and ``forget'' the previous Markov state, and we only need a constant amount of filler bits to break the dependency on the input bits. On the other hand, there are details that need to be worked out for a precise proof (which would differ depending on the exact channel model). For example, after just the initial block of the filler input bits, we break the initial state (distribution) of the channel sequence.
One would need to prove (possibly under additional assumptions for different channel models) that one can achieve the capacity even starting from this ``broken'' state distribution. We leave such a study, which would address, e.g., the capacity of nanopore sequencing models, as future work. 

After establishing the information stability of the Markov-IDS channels, we also provide some upper bounds on the channel capacity for a particular Markov deletion channel model. To do this, we employ the basic methodology developed in \cite{fertonani11bounds} by providing the transmitter and/or the receiver with suitable side information. Our numerical evaluations of the bounds indicate that having memory in the deletion process is beneficial in the sense that it increases the channel capacity for a given deletion probability (more precisely, upper bounds on the channel capacity, which are obtained in the same manner for both scenarios). In other words, we can quantify the expected rate increase for synchronization error channels with memory compared to those with i.i.d. synchronization errors.

The paper is organized as follows.
In Section~\ref{s:prelim}, notations and information-theoretic definitions are introduced.
In Section~\ref{s:funcseq}, some helpful propositions about properties of functions of channels are formulated, which are used both in Dobrushin's proofs and in the proofs in Section~\ref{s:proofs}.
In Section~\ref{s:mark}, the Markov-IDS channel sequence is defined, and the capacity theorems are presented.
In Section~\ref{s:proofs}, the proofs of the capacity theorems are given.
In Section~\ref{s:capa}, the numerical upper bounds on the capacity for the Markov-IDS channel sequence are compared to those of the binary deletion channel.
In Section~\ref{s:conc}, conclusions are drawn.
In Appendix~\ref{s:funcapp}, the proofs of propositions from Section~\ref{s:funcseq} are given.
In Appendix~\ref{s:pert}, the variance of mutual information is analyzed via perturbation of the optimal input distribution.

\section{Preliminaries}
\label{s:prelim}

A mathematical model of a physical channel is defined by conditional distributions $W(y|x)$ for all $x\in \mX, y\in\mY$, which means that if a symbol $x\in\mX$ is input, output $y\in\mY$ is observed with probability\footnote{This is valid only for a countable $\mY$; for the uncountable case, there should be output PDFs instead; in this paper we consider finite $\mX$ and countable $\mY$.} $W(y|x)$.
The fact that $W$ is a channel with the input set $\mX$, and the output set $\mY$ will be denoted by $W:\mX\tosq\mY$. Shannon's theorem \cite{shannon48mathematical} considers a memoryless channel $W:\mX\tosq\mY$ that is used $n$ times independently.
The equivalent channel that transmits $n$ symbols simultaneously is the channel $W^n:\mX^n\tosq\mY^n$, where the conditional probabilities are given by $W^n(y_1^n|x_1^n)=\prod_{i=1}^n W(y_i|x_i)$.

In this paper, our focus is on channels with synchronization errors, modeled as insertions, deletions, and channel errors. As an example, consider an insertion channel $V_n$ with $n$ input symbols $x_1^n\in\mX^n$.
In general, for each input \emph{symbol} $x_i\in\mX$, an output \emph{vector} $\vc{y}_i$ over $\mX$ of length $\geq 1$ is produced. In contrast with the channel $W^n$, the resulting output of $V_n$ is a concatenation of output vectors, but not the vector $\vc{y}_1^n$ of output vectors. For example, if the input bits are $(1,0,0)$ and the corresponding outputs for each bit are $(1,0)$, $(0)$ and $(1,0)$, then the overall channel output would be $(1,0,0,1,0)$, not $((1,0), (0), (1,0))$.

In general, a communication setup, besides conditional probabilities of the original channel, also includes some underlying method of combining the outputs of multiple channel uses. 
In the case of simple memoryless channels, the combining method is just stacking the output symbols into an output vector; in the case of a channel with synchronization errors, the combining method is concatenation. 
For this reason, in the non-trivial setup a general case of \emph{a channel sequence} $(W_n:\mX_n\tosq\mY_n)_{n\in\bN}$ is usually considered \cite{hu62shannon, verdu94general}, where sets $\mX_n$, $\mY_n$ can vary for different $n$.
For simplicity of exposition, we consider the case when the input set for $W_n$ is $\mX_n=\mX^n$ and $\mX$ is fixed for all $n$.

For a channel sequence $W=(W_1,W_2,\dots)$, the \emph{mutual information capacity} (or i-capacity) is defined as 
\begin{align}
\bfI(W)=\lim_{n\to\infty}\bfI(W_n)/n,
\end{align} 
if the limit exists, where $\bfI(W_n)$ denotes the maximum mutual information between the input and the output of a channel $W_n$ over all possible input distributions. 
The \emph{coding capacity (c-capacity)} is the (maximum) asymptotic rate of coding schemes that achieve arbitrary small error probabilities as $n\to\infty$. The \emph{capacity theorem} is a statement that the i-capacity is equal to the c-capacity.

There are many capacity theorems that generalize Shannon's theorem on non-trivial channel sequences.
Dobrushin's work \cite{dobrushin59ageneral} considers the most general case, where input and output alphabets can be continuous, and the metric of the decoder output's closeness to the original signal can be arbitrary.
It is shown that the so-called ``\emph{information stability}'' of a channel sequence is sufficient for the capacity theorem to hold.
Later, in \cite{hu62shannon}, it was shown that information stability is also necessary for the capacity theorem.
In \cite{dobrushin67shannons}, Dobrushin proves the capacity theorem for channels with i.i.d. synchronization errors where an output vector of finite expected length is produced for each input symbol, and all the output vectors are concatenated to form the resulting output. 
The generalization of the capacity theorem to the case of IDS channels with continuous input alphabets is studied in~\cite{stambler70memoryless}.
Capacity theorems are also proven for FSMC \cite{gallager68information}, for Gilbert-Elliott channels \cite{mushkin89capacity}, for deletion channels concatenated with FSMC \cite{li21capacity}, for IDS channels with intersymbol interference \cite{mao18models}.
In \cite{mcbain24noisy}, the Markov-constrained coding theorem is proven for the Markov duplication channel.
In \cite{con2025channels}, Section~1.2, there is a short yet comprehensive literature review on the existing coding theorems.

\subsection{Notation}

Sets are written in calligraphic letters $\mX,\mY,\omY$; operators are written in bold letters $\bfC,\bfE,\bfI$.
Vectors and subvectors are denoted by $a_b^c=(a_b,a_{b+1},\dots,a_c)$, where $b,c\in\bN$.
Also, we denote infinite sequences as $a^\infty=(a_1,a_2,\dots)$.
For a set $\mA$, we denote a set of all finite sequences with elements from $\mA$ by
$
\omA=\set{()}\cup \bigcup_{i=1}^\infty \mA^i.
$
We denote the length of a vector $a=a_1^n\in\omA$ by $\bfN(a)=n$. The concatenation (gluing) operator $\bg$ is defined as 
$
\bg(a_1^n,b_1^k)=(a_1,a_2,...,a_n,b_1,...,b_m),
$
which is also generalized to any number of input vectors.
Also, define the function $\bg_n:\omA^n\to\omA$, which glues $n$ vectors $\oa_i\in\omA$ as
$
\bg_n(\oa_1^n)=\bg(\oa_1,\oa_2,...,\oa_n).
$
The inverse image $\bg^{-1}_n(a)$ is the set of all possible ways of cutting $a\in\omA$ into $n$ subvectors, i.e., $\bg^{-1}_n(a)=\set{\oa_1^n|\bg_n(\oa_1^n)=a}$.
For a set $\mS$, we denote by $\bD_\mS$ the set of all distributions over $\mS$.
For $n\in\bN$, denote by $\bD_n$ the set of distributions over the set $\set{1,2,\dots,n}$ in a form of row vectors $\rho_1^n$, $\sum_{i=1}^n\rho_i=1$, $\rho_i\geq 0$.


We also note that in the literature it is often not obvious which probability distribution is assumed for probabilities or expectations. Therefore, to remove any ambiguities, we explicitly specify the probability distribution for probabilities, expectations, and entropy expressions as a subscript of $\Pr$, $\bfE$, $\bfH$.
The values that are not known are interpreted as random variables, and we always take summation/expectation over all their possible values.
For example, if we have a distribution $P(x,y)$, where $x\in\mX,y\in\mY$, and $y_0 \in \mY$ is a fixed value, then
\begin{align}
\Pr_{P(x,y)}\brsq{x>0}&=\sum_{\substack{x\in\mX:x>0\\y\in\mY}}P(x,y), 
\label{eq:prex}
\\
\bfE_{P(x,y_0)}\brsq{f(x)}&=\sum_{x\in\mX}P(x,y_0)\cdot f(x).
\label{eq:eex}
\end{align}
We write $\bfH_A\brsq{B}=\bfE_A\brsq{-\log_2B}$ for entropy-like expressions.
If $A=B$, then we omit the subscript for simplicity:
\begin{align}
\bfH\brsq{W(x,y)}
&=-\sum_{x\in\mX,y\in\mY}W(x,y)\cdot \log_2 W(x,y).
\label{eq:hex2}
\end{align}

We use $O$- and $o$- notation for sequences / functions:
\begin{align}
a_n\in o(b_n) &\iff \lim_{n\to\infty}\frac{a_n}{b_n}=0,
\\
a_n\in O(b_n) &\iff \lim_{n\to\infty}\frac{a_n}{b_n}\in\bR.
\end{align}
Note that $o(f(n)) \subseteq O(f(n))$.

\subsection{Channels and Channel Sequences \label{s:chseq}}

Recall (see the beginning of Section~\ref{s:prelim}) that a channel $W:\mX\tosq\mY$ is a collection of conditional distributions $W(y|x)$, defined for all $x\in\mX$, $y\in\mY$.
Throughout the paper, we assume \emph{finite $\mX$ and countable $\mY$}.

If we further define some input distribution $p\in\bD_\mX$, then we can compute the joint distribution of input and output and the distribution of output:
\begin{subequations}
\begin{align}
\pWd{p}{W}(x,y)&=p(x)W(y|x)
\label{eq:pWxy}
\\
\pWyd{p}{W}(y)&=\sum_{x\in\mX}\pWd{p}{W}(x,y).
\label{eq:pWy}
\end{align}
\end{subequations}

For channels $W_i:\mX_i\tosq\mY_i$, define the channel product $W=\prod_{i=1}^n W_i$ as a channel $W:\mX_1\times\dots\times\mX_n\tosq\mY_1\times\dots\times\mY_n$, such that
\begin{align}
W=\prod_{i=1}^n W_i \iff W(y_1^n|x_1^n)=\prod_{i=1}^n W_i(y_i|x_i) 
\label{eq:prodW}
\end{align}
A similar notation is used for the $n$-th power $W^n:\mX^n\tosq\mY^n$ of a channel $W:\mX\tosq\mY$.

A \emph{channel sequence} $V$ is defined as an infinite sequence $V^\infty$ of channels $V_i:\mX^i\tosq\mY_i$ for some fixed set $\mX$.
Note that channel sequence is not just any infinite sequence of channels: the $i$-th channel should transmit exactly $i$ symbols from $\mX$.

\subsection{Mutual Information Capacity (i-capacity)}

The mutual information density of a channel $W$ for an input distribution $p\in\bD_\mX$ and given values of an input $x\in\mX$ and an output $y\in\mY$ is defined as:
\begin{align}
\bi(W,p,x,y)\stackrel{\triangle}{=}
\begin{cases}	
	0, \text{if } W(y|x)=0 \text{ or } \pWyd{p}{W}(y)=0 \\	\log_2\displaystyle\frac{W(y|x)}{\pWyd{p}{W}(y)}, \text{ otherwise}
\end{cases}
\label{eq:IWpxy}
\end{align}
The mutual information of $W$ under input distribution $p$ is the expectation of mutual information density:
\begin{align}
\bfI(W,p)\;\defeq \; &
\bfE_{\pW{p}{W}(x,y)}[\bi(W,p,x,y)]
\nonumber\\
\;=\; &\sum_{x\in\mX,y\in\mY}\pWd{p}{W}(x,y)\cdot \bi(W,p,x,y).
\label{eq:IWp}
\end{align}
The \emph{mutual information capacity (i-capacity)} \cite{dobrushin59ageneral} of a \emph{channel} $W$ is the maximum possible mutual information of $W$ under any input distribution, and is given by%
\begin{align}
\bfI(W)\ \defeq\ \max_{p\in\bD_\mX}\bfI(W,p).
\label{eq:IWn}
\end{align}
We will call the distribution $p$, for which the maximum is achieved, the \emph{optimal distribution} (for $W$).
By the \emph{i-capacity} of \emph{channel sequence} $V=V^\infty$ we mean the asymptotic average i-capacity of $V_n$ per input symbol:
\begin{align}
\bfI(V)\ \defeq\ \lim_{n\to\infty}\frac{\bfI(V_n)}{n},
\label{eq:IW}
\end{align}
if the limit exists. Note that the i-capacity of a \emph{channel} always exists, since $\bfI(W,p)$ is a concave function of the input distribution $p$ with a countable number of coefficients; however, the i-capacity of a \emph{channel sequence} might not exist, since sequence $\bfI(V_n)/n$ might not converge%
\footnote{Consider a channel which is ideal for even $n$ and complete noise for odd $n$.
In this case, the sequence of values of $\bfI(V_n)/n$ is $\log_2|\mX|$, $0$, $\log_2|\mX|$, $0,\dots$
Such sequence has no limit for $|\mX|>1$.}.

For a channel $W$ and its input distribution $p$, define the \emph{mutual information variance} as the variance of mutual information density:
\begin{align}
\bfV(W,p)=\underset{\pW{p}{W}(x,y)}{\bfE}\brsq{\br{\bfi(W,p,x,y)-\bfI(W,p)}^2}.
\label{eq:vardef}
\end{align}

\subsection{Coding Capacity (C-capacity)}

For a channel $W:\mX\tosq \mY$, code size $M\in\bN$, $M\leq |\mX|$ and error probability $\ve\in[0,1]$, we define a \emph{$(W, M, \ve)$-coding scheme} as a collection $(x_1^M, \mR_1^M)$, where:
\begin{itemize}	
	\item all $x_i$ are distinct codewords from $\mX$, $i=1\dots M$;
	\item $\mR_1^M$ is a collection of $M$ disjoint sets $\mR_i\subseteq \mY$, which are the decoder's decision regions. When $y\in\mR_i$ is received, the decoder outputs $x_i$;
	\item the decoding error probability is not greater than $\ve$ for any $x_i$, i.e.
\begin{align}
	\forall i\in\set{1,\dots,M}: \sum_{y\notin\mR_i}{W(y|x_i)}\leq \ve.
\label{eq:decerr}
\end{align}
\end{itemize}

Consider a channel sequence $V^\infty$.
Define the set $\bfR(V)$ of all achievable rates over $V$ as the set of $R\in\bR_{\geq 0}$, such that
\begin{align}
\forall \ve>0, \exists n_0:\forall n\geq n_0: \exists (V_n, \ceil{2^{nR}}, \ve)\textit{-coding scheme}.
\label{eq:ar}
\end{align}
By \textit{coding capacity (c-capacity)} $\bfC(W)$ of the \emph{channel sequence} $V$ we mean the supremum of achievable rates \cite{dobrushin67shannons}:
\begin{align}
\bfC(V)&\stackrel{\triangle}{=}\sup\bfR(V),
\label{eq:cdef}
\end{align}
The existence of a $(W,M,\ve)$-coding scheme obviously implies the existence of a $(W,M',\ve')$-coding scheme for any $M'\leq M$ and $\ve'\geq\ve$.
So, if $R\in\bfR(V)$, then $R'\in\bfR(V)$ for any $R'<R$.
Thus, the set $\bfR(V)$ is a connected set, i.e., $\bfR(V)\in\set{[0,\alpha], [0,\alpha)}$ for some $\alpha$,
and therefore any channel sequence has its c-capacity.

In~\cite{verdu94general}, the c-capacity $\bfC(V)$ is derived in terms of mutual information density: $\bfC(V)$ is the supremum of all $\alpha$, for which
\begin{align}
\lim_{n\to\infty}~\min_{p_n\in\bD_{\mX^n}}\Pr_{\pW{p_n}{V_n}(x,y)}\set{\frac{\bi(V_n,p_n,x,y)}{n}\leq \alpha}=0.
\label{eq:cformula}
\end{align}

\subsection{Information Stability}

A channel sequence $V=V^\infty$ is called \emph{$J$-information stable} \cite{hu62shannon}, where $J=J^\infty$ with $J_n\to\infty$, if there exist distributions $p^*_n\in\bD_{\mX^n}$ and sequence $\delta^\infty$ with $\delta_n>0$, $\delta_n\to 0$, such that
\begin{align}
\Pr_{\pW{p^*_n}{V_n} V_n(x,y)}\brsq{\abs{\frac{\bi(V_n,p^*_n,x,y)}{J_n}-1}>\delta_n} <\delta_n.
\label{eq:iso}
\end{align}
If $J_n=\bfI(V_n)$ then $V$ is called information-stable.
If, furthermore, the limit $\bfI(V)$ of $\bfI(V_n)/n$ exists, then we call $V$ a strictly information-stable channel sequence and write $\bfS(V)$ to denote this fact.
In the case of $\bfS(V)$, under our general assumption that $\mX$ is finite, \eqref{eq:iso} is equivalent to the following: for all $\delta>0$,
\begin{align}
\lim_{n\to\infty}~\Pr_{\pW{p^*_n}{V_n}(x,y)} \brsq{\abs{\frac{\bi(V_n,p^*_n,x,y)}{n}-\bfI(V)}>\delta} = 0.
\label{eq:is}
\end{align}

\begin{proposition}[Theorems~3.1 and 3.2 in \cite{hu62shannon}]
\label{p:is}
$\bfS(V)$ iff $\bfC(V)=\bfI(V)$.
\end{proposition}

For further derivations, we need the following sufficient condition of information stability.
\begin{proposition}[A sufficient condition for strict information stability]
\label{p:ise}
If there exists $\bfI(V)$ and a sequence of distributions $p^*$, for which
\begin{align}
\lim_{n\to\infty}~\bfE_{\pW{p^*_n}{V_n}(x,y)}\brsq{\abs{\frac{\bi(V_n,p^*_n,x,y)}{n}-\bfI(V)}} = 0,
\label{eq:ise}
\end{align}
then $\bfS(V)$.
\end{proposition}
\begin{proof}
Note that for any $\delta>0$,
\begin{align}
&0  \leq\lim_{n\to\infty}~\Pr\brsq{\abs{\frac{\bi(V_n,p^*_n,x,y)}{n}-\bfI(V)}>\delta}
\nonumber\\
&\leq\frac{1}{\delta}\cdot \lim_{n\to\infty}~\bfE\brsq{\abs{\frac{\bi(V_n,p^*_n,x,y)}{n}-\bfI(V)}}=0,
\label{eq:p2proof}
\end{align}
which is exactly the information stability condition \eqref{eq:is}.
\end{proof}


\section{Functions of Channel Sequences and Their Properties}
\label{s:funcseq}

In this Section, we provide several propositions employed in the proofs of the main theorems. 
The propositions are dedicated to sequences of functions, which are applied to the outputs of each channel in a channel sequence.
Namely, they describe two properties, which are sufficient for function sequences not to change the channel sequence properties, such as i-capacity, c-capacity and $\bfS$ property.
Note that here, in contrast to the next Section, the channel sequence is arbitrary.

The first property (Proposition~\ref{p:invimage}) is that the asymptotic size of the inverse image of functions is small.
This means that the functions do not have the same value for sufficiently large number of arguments, or in other words, do not ``stick together'' many channel's outputs.
The second property (Proposition~\ref{p:pr0}) is that the functions are bijective with high probability, meaning that the channel outputs, for which we cannot inverse the function, are  improbable in the asymptotic regime.

While some of the results are new, others are restatements of Dobrushin's results to make them modular and more accessible.
Also, we would like to highlight that we tried to make these propositions reusable by others when proving capacity theorems.
The proofs are available in the Appendix~\ref{s:funcapp}.

For a channel $W:\mX\tosq\mY$ and function $f:\mY\to\mZ$ we write $W'=f(W)$ to refer to a channel $W':\mX\tosq\mZ$, defined as a composition of channel $W$ and function $f$.
That is,
\begin{align}
W'(z|x)=\sum_{y\in f^{-1}(z)}W(y|x),
\label{eq:fw}
\end{align}
where $f^{-1}(z)=\set{y:f(y)=z}$ is the inverse image of $f$.
Channel $W'$ is also called \emph{degraded with respect to} $W$, and it is well-known that $\bfC(W')\leq\bfC(W)$, and if both $\bfI(W)$ and $\bfI(W')$ exist, then $\bfI(W') \leq \bfI(W)$.

Let $\Phi=\log_2\max_{z\in \mZ}\abs{f^{-1}(z)}$ be the maximum log-size of a preimage.
Then, the following proposition holds.

\begin{proposition}[Equation (4.3) in \cite{dobrushin67shannons}]
\label{p:preimage}
For any input distribution $p\in\bD_{\mX}$:
\begin{align}
\underset{\pW{p}{W}(x,y)}{\bfE}\bigg[\brabs{\bi(W,p,x,y) - \bi(W',p,x,f(y))}\bigg]\leq \Phi
\label{eq:absi-i}
\end{align}
\end{proposition}

In what follows, consider a channel sequence $U_n:\mX^n\tosq\mY_n$ and a sequence of functions $f_n:\mY_n\to\mZ_n$.
Denote by $V=f(U)$ a channel sequence $V_n=f_n(U_n):\mX^n\tosq\mZ_n$.
Also, denote by $\Phi_n=\log_2\max_{z\in \mZ_n}\abs{f_n^{-1}(z)}$ the logarithm of maximum preimage size.

\begin{proposition}[Function with small preimage]
\label{p:invimage}
If $\Phi_n\in o(n)$, then $\bfC(U)=\bfC(V)$.
If, furthermore, at least one of $\bfI(U)$, $\bfI(V)$ exist, then they both exist and $\bfI(U)=\bfI(V)$.
Moreover, $\bfS(U)\iff\bfS(V)$.
\end{proposition}

Thus, one can apply such function sequences to a channel sequence $W$, changing neither $\bfI(W)$ or $\bfC(W)$, or information stability property.

\begin{proposition}[Semi-bijective function]
\label{p:prb}
Consider channels $U:\mX\tosq\mY$ and $V:\mX\tosq\mZ$, such that $V=f(U)$, $f:\mY\to\mZ$.
Consider partitions $\mY=\mA \sqcup \mB$, $\mZ=\mA' \sqcup \mB'$, such that $f(\mA)=\mA'$ and $f$ is a bijection between $\mA$ and $\mA'$.
Denote by $\beta(x)=\Pr_{U(y|x)}\brsq{y\in\mB}$ the probability that the output of $U$ belongs to $\mB$ given input $x\in\mX$.
If $\forall x\in\mX: \beta(x)\leq\obeta$, then for any $p\in\bD_\mX$,
\begin{align}
\bfI(U,p)-\bfI(V,p)\leq \obeta \log_2|\mX|+\frac{1}{e\ln 2}.
\label{eq:uvbeta}
\end{align}
\end{proposition}

Obviously, renaming the output symbols does not change any properties of a channel.
Thus, if $V=f(U)$ and the function $f$, applied to the output of the channel $U$, is bijective, then $V$ is virtually the same channel as $U$.
Roughly speaking, Proposition~\ref{p:prb} states that if the output with high probability falls into the set $\mA$, where the function $f$ is bijective, then the change of the i-capacity is small.

\begin{proposition}[Function that does nothing almost surely]
\label{p:pr0}
Consider channel sequences $U_n:\mX^n\tosq\mY_n$ and $V_n:\mX^n\tosq\mZ_n$, such that $V_n=f_n(U_n)$, $f_n:\mY_n\to\mZ_n$.
Consider partitions $\mY_n=\mA_n \sqcup \mB_n$, $\mZ_n=\mA'_n \sqcup \mB'_n$, such that $f_n(\mA_n)=\mA'_n$ and $f_n$ is bijective between $\mA_n$ and $\mA'_n$.
Denote by $\beta_n(x)=\Pr_{U_n(y|x)}\brsq{y\in\mB_n}$, and $\obeta_n=\max_x\beta_n(x)$.
If the output of channel $U_n$ almost always belongs to $\mA_n$, i.e., $\lim_{n\to\infty}\obeta_n=0$, then $\bfC(U)=\bfC(V)$.
If, in addition, one of $\bfI(U)$, $\bfI(V)$ exists, then they both exist and $\bfI(U)=\bfI(V)$.
Moreover, $\bfS(U)\iff\bfS(V)$.
\end{proposition}

If this proposition is combined with Proposition~\ref{p:invimage}, the result is that one can apply any finite number of functions, which either have sufficiently small pre-image or sufficiently small probability of changing the output, without changing neither $\bfI$ nor $\bfC$ of the original channel sequence.
The $\bfS$ property is also preserved upon such transformations.

\section{Markov-IDS Channels and the Main Theorems}
\label{s:mark}

In this Section we introduce step-by-step a specific family of channel sequences, namely, the Markov-IDS channel sequence. 
After the definitions, we formulate the main capacity theorem, which is split into two theorems, which are given at the end of this Section. 
The proofs are placed in the next Section, and they will need the propositions, formulated for an arbitrary channel sequence in Section~\ref{s:funcseq}.

\subsection{The IDS Channel Sequence}

Consider an IDS channel $W:\mX\tosq\omY$ as in \cite{dobrushin67shannons}, such that for each $x\in \mX$ the expected length $\bfN(y)$ of the output sequence $y\in\omY$ is finite:
\begin{align}
\exists A\in\bR: \forall x \in \mX: \bfE_{W(y|x)}\brsq{\bfN(y)}\leq A.
\label{eq:eN}
\end{align}
The \emph{IDS channel} $D_n$ is defined as channel $W^n$, followed by a gluing operator on its output (see definition \eqref{eq:fw}):
\begin{align}
D_n=\bg_n(W^n),
\label{eq:wids}
\end{align}
where $\bg_n:\omY^n\to\omY$ glues together $n$ vectors over $\mY$.
The IDS channel sequence $D^\infty$ is information-stable \cite{dobrushin67shannons}, so $\bfC(D)=\bfI(D)$.

\begin{proposition}
\label{p:finH}
The entropy of the output length of IDS channel (with finite expectation of output length) is finite:
\begin{align}
\exists \oh\in\bR: \forall p\in\bD_\mX, \bfH_{[pW](x,y)}\brsq{\bfN(y)} \leq \oh.
\label{eq:finlenH}
\end{align}
The conditional output entropy for each input is finite as well:
\begin{align}
\exists \oH\in\bR: \forall p\in\bD_\mX, \bfH_{\pW{p}{W}(x,y)}\brsq{W(y|x)} \leq \oH.
\label{eq:finH}
\end{align}
Moreover, the unconditional output entropy is also finite:
\begin{align}
\exists \oH'\in\bR: \forall p\in\bD_\mX, \bfH\brsq{\pW{p}{W}(y)} \leq \oH'.
\label{eq:finuncH}
\end{align}
\end{proposition}
\begin{proof}
Fix some input distribution over $\mX$.
The random variable $\bfN(y)$ is distributed over $\bZ_{\geq 0}$, so its entropy is not higher than that of the max-entropy distribution over $\bZ_{\geq 0}$, which is the geometric distribution with finite entropy \cite{lisman72note}.
Denote the latter by $\oh$ and obtain \eqref{eq:finlenH}.

Fix some $x_0\in\mX$.
Denote by $w_y=W(y|x_0)$.
The conditional entropy $\bfH(W(y|x_0))$ is a series 
\begin{align}
\bfH\brsq{w_y}=\sum_{y\in\omY:w_y>0}-w_y\log_2w_y.
\label{eq:Elog}
\end{align}
Since all the terms of the summation are non-negative, we only need to prove that the sum is finite.
Denote by $P_n=\sum_{y\in\mY^n}w_y$ the probability that the output length is $n$.
From~\eqref{eq:eN}, 
\begin{align}
\sum_{y\in\omY}w_y\cdot\bfN(y)=\sum_{n=0}^\infty n\cdot \sum_{y\in\mY^n} w_y = \sum_{n=0}^\infty n\cdot P_n\leq A.
\label{eq:wyNyleqA}
\end{align}
Expression $\sum_{y\in\mY^n}(-w_y\log_2w_y)$ achieves its maximum (with fixed $P_n$), when $w_y=P_n/|\mY|^n$.
Thus,
\begin{align}
\bfH\brsq{w_y}
\leq -\sum_{n=0}^\infty |\mY|^n\cdot\frac{P_n}{|\mY|^n}\log_2\frac{P_n}{|\mY|^n}
\leq \bfH\brsq{P_n}-A \log_2 |\mY|.
\label{eq:p3proof}
\end{align}
The first term is finite by \eqref{eq:finlenH}, the second term is finite by the definition of $A$ and $|\mY|$. 
Denote by $\oH=\max_x \bfH(W(y|x))$.
For any input distribution $p\in\bD_\mX$, the conditional entropy $\bfH_{\pW{p}{W}(x,y)}\brsq{W(y|x)}=\sum_x p(x)\bfH\brsq{W(y|x)} \leq \oH$.

The unconditional entropy can be bounded using the formula $H(Y)\leq H(X,Y) = H(X)+H(Y|X)$, where $H(Y|X)$ is already upper-bounded by $\oH$, and $H(X)$ is the entropy of the input symbol, which does not exceed $\log_2|\mX|$. 
Letting $\oH'=\oH+\log_2|\mX|$ we obtain \eqref{eq:finuncH}.
\end{proof}

\subsection{The Markov-IDS Channel Sequence}

Consider a discrete-time Markov chain with $s$ states $\mS=\set{1,\dots,s}$. 
Let $G$ be an $s\times s$ matrix which defines the transition probabilities: $G_{\sigma, \tau}\in[0,1]$ is equal to the conditional probability of the next state being $\tau$ if the current state is $\sigma$.

Assume that each state $\sigma\in\mS$ corresponds to a \emph{state channel} $W_1[\sigma]:\mX\tosq\omY$, which is an IDS channel\footnote{If the state channels have different output alphabets $\mY_\sigma$ for each $\sigma$, we can easily generalize this case by letting $\mY=\cup_\sigma \mY_\sigma$.}.

The FSMC channel \cite{gallager68information} works as follows:
\begin{enumerate}
	\item The initial state $\sigma_1$ is given.
	\item The $i$-th symbol $x_i\in\mX$ is transmitted through channel $W_1[\sigma_i]$. 
	The output is $\oy_i\in\omY$.
	\item After the $i$-th symbol is transmitted, the system transits from state $\sigma_i$ to state $\sigma_{i+1}$ with probability $G_{\sigma_i,\sigma_{i+1}}$.	
	\item For $n$ input symbols, one obtains $\oy_1^n\in\omY^n$. 	
\end{enumerate}

Let $W_{*,1}[\sigma_1]=(W_{1,1}[\sigma_1],W_{2,1}[\sigma_1],\dots)$ be an FSMC channel sequence with initial state $\sigma_1$.
Then
\begin{align}
&W_{n,1}[\sigma_1](y_1^n|x_1^n)=
W_1[\sigma_1](y_1|x_1)
\nonumber\\&\times
\sum_{\sigma_2^n\in\mS^{n-1}} \prod_{i=2}^n G_{\sigma_{i-1},\sigma_i}W_1[\sigma_i](y_i|x_i).
\label{eq:Wn1}
\end{align}
Introduce channel $W_n[\sigma]$ which outputs the glued output of $W_{n,1}[\sigma]$:
\begin{align}
W_n[\sigma]=\bg_n(W_{n,1}[\sigma]).
\label{eq:Wn}
\end{align}
Assume that the initial state of the Markov chain is chosen according to distribution $\rho \in\bD_s$.
Then, the corresponding FSMC channel sequence is denoted by $V_{n,1}[\rho]: \mX^n\tosq\omY^n$:
\begin{align}
V_{n,1}[\rho](y|x)=\sum_{\sigma\in\mS}\rho_\sigma \cdot W_{n,1}[\sigma](y|x).
\label{eq:Vn1}
\end{align}
\emph{The Markov-IDS channel sequence} $V_n[\rho]=\bg_n(V_{n,1}[\rho])$ is given by gluing the output of the FSMC channel:
\begin{align}
V_n[\rho](y|x)=\sum_{\sigma\in\mS}\rho_\sigma \cdot W_n[\sigma](y|x). 
\label{eq:Vn}
\end{align}

Denote by $V_{n_1|n_2|...|n_t}[\rho]:\mX^n\tosq \omY^t$, $n=\sum_i n_i$, a channel, which works as a channel $V_{n,1}[\rho]$ combined with the merging operation.
The merging operation concatenates the outputs inside $i$-th output sub-block for each $i$, where the corresponding $i$-th input sub-block has length $n_i$.
Therefore, the receiver knows the output corresponding to each $n_i$ block but does not know the exact output for each input symbol.
Also, we write $\tn_i$ instead of $n_i$ to denote the fact that the output, corresponding to the $i$-th input block, is erased, i.e., substituted by the zero-length vector.

\begin{proposition}[Entropy of a separate output block of Markov-IDS channel is linear]
\label{p:linH}
Consider a channel $\tW=V_{n_1|n_2|n_3}[\rho](y_1^3|x_1^3)$, let $X_i$ be a random variable, corresponding to $x_i$, i.e., the $i$-th input block of $n_i$ symbols, and $Y_i$ be a random variable of the output $y_i$, corresponding to input $x_i$. 
Then, the unconditional entropy of $Y_2$ is at most linear in $n_2$:
\begin{align}
\exists \tH\in\bR: \forall p\in\bD_{\mX^{n}}, \bfH_{\pW{p}{\tW}(x_1^3,y_1^3)}\brsq{\pW{p}{\tW}(y_2)} \leq \tH n_2,
\label{eq:finHmarkov}
\end{align}
where $n=n_1+n_2+n_3$ and $\tH$ does not depend on $n_i$ nor on $\rho$.
Moreover, the output entropy of $V_n[\rho]$ does not exceed $\tH n$.
\end{proposition}
\begin{proof}
By Proposition~\ref{p:finH}, the entropy of the output of each component IDS channel $W[\sigma]$ is upper-bounded by a constant $\oH'_\sigma$, and the entropy of the output \emph{length} of $W[\sigma]$ can be upper-bounded by $\oh_\sigma$.
Let $\oH'=\max_\sigma \oH'_\sigma$, $\oh=\max_\sigma \oh_\sigma$.

Consider telling the following additional information (apart from $y_1^3$) to the receiver: 1) the $n_2$ states the Markov chain was in during each transmitted bit for the middle $n_2$-block; 2) the length $t_i\in\bZ_{\geq 0}$ of output for each input bit, $i=1\dots n_2$.
It is obvious that the amount of the additional information is not greater than $n_2(\log_2s + \oh)$ bits.
The output $y_2$ thus can be split into $n_2$ parts $z_1^{n_2}$, $z_i\in \omY$, which correspond to the output of a component channel for each input bit.
Now, for each input bit in the $n_2$-block, we know the output $z_i$ and the state channel $W[\sigma]$, which transmitted the symbol $(x_2)_i$.
Thus, the (unconditional) entropy of $z_i$ can be upper-bounded by $\oH'$. 
The entropy of $y_2$ can be upper-bounded by the sum of entropies of $z_i$ plus the amount of the additional information, which we have told the receiver:
\begin{align}
\bfH_{\pW{p}{\tW}(x_1^3,y_1^3)}\brsq{\pW{p}{\tW}(y_2)} \leq n_2(\oH'+\oh+\log_2 s).
\label{eq:finHmarkov1}
\end{align}
Letting $\tH=\oH'+\oh+\log_2 s$, we obtain \eqref{eq:finHmarkov}. 
Recall that $\tH$ only depends on the set of state channels $W[\sigma]$, $\sigma\in\mS$.
Letting $n_1=n_3=0$, we can obtain the additional statement of this proposition.
\end{proof}

\subsection{Markov Chain State Distribution Convergence}
Throughout the section, we use notation of the channel sequences $W_{n,1}[\sigma]$, $W_n[\sigma]$, $V_{n,1}[\rho]$ and $V_n[\rho]$, $n\in\bN$, as defined in \eqref{eq:Wn1}, \eqref{eq:Wn}, \eqref{eq:Vn1} and \eqref{eq:Vn}.

It is well-known (Th. 8.9 in~\cite{billingsley12probability}), that if a Markov chain with a finite state space is aperiodic and irreducible, then there exist constants $B,B'\geq 0$, $0\leq P<1$ and $C>0$, such that
\begin{align}
\forall n\in\bN,\forall i,j\in\mS: |G^n_{i, j}-\pi_j|\leq BP^n=B'e^{-Cn}\;\defeq\;\delta_n,
\label{eq:Gn-pi}
\end{align}
where $G^n_{i,j}=(G^n)_{i,j}$, and $\pi\!\in\!\bD_s$ is the stationary distribution of the Markov chain.
Moreover, for any initial distribution $\rho$,
\begin{align}
\brabs{(\rho G^n)_j-\pi_j}
\leq \sum_i\rho_i\brabs{G^n_{i,j}-\pi_j}
\stackrel{\eqref{eq:Gn-pi}}{\leq} \delta_n.
\label{eq:rhoGn-pi}
\end{align}
Introduce sequence $\ve_n=\delta_n/\min_i \pi_i$.
Then, the small additive deviation (with $n\to\infty$) can be replaced with a small multiplicative deviation as $\pi_i+\delta_n\leq \pi_i(1+\ve_n)$. 
Combining this with~\eqref{eq:Gn-pi} and \eqref{eq:rhoGn-pi},
\begin{align}
\forall \rho\in\bD_s, \forall j:~&\pi_j(1-\ve_n)\leq (\rho G^n)_j \leq \pi_j(1+\ve_n),
\label{eq:rhogn-as}\\
&\ve_n\leq De^{-Cn}, C=\const,D=\const,
\label{eq:ve-as}
\end{align}
which also implies 
\begin{align}
\pi_j(1-\ve_n)\leq G^n_{i, j} \leq \pi_j(1+\ve_n).
\label{eq:gijdounded}
\end{align}
\subsection{Main Theorems}

\begin{theorem}[Existence of i-capacity]
\label{theorem1}
If a Markov chain $G$ of $s$ states, corresponding to Markov-IDS channel sequence $V$, is aperiodic and irreducible, then for any $\rho\in\bD_s$ the i-capacity $\bfI(V[\rho])$ exists and does not depend on $\rho$.
In other words, for any $\rho \in \bD_s$, $\bfI(V[\rho])$ exists and
\begin{align}
\bfI(V[\rho])=\bfI(V[\pi])=\bfI(V),
\label{eq:t1}
\end{align}
where $\pi$ is the stationary distribution of the Markov chain.
\end{theorem}

\begin{theorem}[Information stability of the Markov-IDS channel sequence]
\label{t:cviv}
Assume that the length of the output of all state channels $W[\sigma]$ is constrained, i.e., $W[\sigma]:\mX\tosq \tmY$, where 
\vspace*{-0.2cm}
\begin{align}
\tmY=\cup_{i=0}^A \mY^i,
\label{eq:tmY}
\end{align}
and $A$ is a constant which does not depend on $\sigma$.

Then, for any $\rho$, the channel sequence $V[\rho]$ is information stable.
In particular, the c-capacity $\bfC(V[\rho])$ achieves the i-capacity $\bfI(V)$, and thus does not depend on $\rho$:
\begin{align}
\bfC(V[\rho])=\bfC(V)=\bfI(V).
\label{eq:vnisis}
\end{align}
\end{theorem}

Note that for Theorem~\ref{theorem1}  we do not need the assumption of finite-length output of the component channels. 
This is achieved by employing Propositions~\ref{p:finH} and \ref{p:linH}.
Theorem~\ref{t:cviv} might also avoid this assumption, but we did not work on the generalization of the proof.
On the other hand, if we had the finite-length assumption in Theorem~\ref{theorem1}, then the proof would be simpler, since we would not need the Propositions~\ref{p:finH} and \ref{p:linH}.

The proofs of these theorems are given in the next section, employing several general results provided in the appendices.

In our development in Appendix~A, we also introduce a unified approach to elucidate the idea behind Dobrushin's methods, which can be seen as \emph{applying functions to channel sequences}. If a function is applied to a channel's output, and this function is not ``doing much'' in a certain sense, then the modified channel has the same i-capacity, c-capacity, and information stability as the original one. This 
means that one can apply any finite number of such functions and still have the same capacities.

In Appendix~B, we provide two results that are also used in proving our theorems: Proposition~\ref{p:peps} provides how perturbation of the sequence of optimal input distributions influences the asymptotic behavior of i-capacity.
Proposition~\ref{p:pvar} gives the order of growth of the mutual information variance.

\section{Proofs of the Main Theorems}
\label{s:proofs}

\subsection{Proof of Theorem \ref{theorem1}}
We first will prove that the limit 
$\bfI(V[\rho])$ exists.
For that, we will use the stronger formulation \cite{bruijn52linear} of Fekete's lemma, which says that if function $u:\bN\to\bR$ is quasi-subadditive, i.e., if for all intergers $m,n$, such that $n \leq m \leq 2n$,
\begin{align}
u(m+n)\leq u(m)+u(n)+f(m+n),
\label{eq:gfl}
\end{align}
where function $f(n)$ is such that the series $\sum_{n=1}^\infty f(n)/n^2$ converges, then the limit $\lim_{n\to\infty}u(n)/n$ exists.

Consider the case of $\rho=\pi$, i.e., the initial state distribution is the stationary distribution of the Markov chain (later, it will become clear that the initial state distribution does not influence the i-capacity).
Then, the following sequence of inequalities holds:
\begin{align}
&\bfI(V_{m+n}[\pi]) \stackrel{1}{\leq}
\bfI(V_{m|n}[\pi]) \stackrel{2}{\leq}
\bfI(V_{m|l|n-l}[\pi]) 
\nonumber\\&
\stackrel{3}{\leq}
\bfI(V_{m|\tl|n-l}[\pi])+l\tH
\stackrel{4}{\leq} 
\bfI(V_{m|\tl|n-l|l}[\pi])+l\tH
\nonumber\\ &
\stackrel{5}{\leq} 
\bfI(V_{m|\tl|n}[\pi])+l(\oh+\tH)
\label{eq:VmnVmln}
\end{align}
where $\oh$ and $\tH$ are constants from Propositions~\ref{p:finH} and \ref{p:linH}.

The transitions $1$ and $2$ correspond to providing the receiver with additional information.

In transition $3$, assume the same input distribution $p\in\bD[\mX^{m+n}]$ for both $V_{m|l|n-l}[\pi]$ and $V_{m|\tl|n-l}[\pi]$.
We establish that the transition $3$ is valid by upper-bounding the difference $\bfI(V_{m|l|n-l}[\pi], p)-\bfI(V_{m|\tl|n-l}[\pi], p)$ for any $p$.
The amount of information that we have lost by erasing the $l$-block is upper-bounded by $l\tH$, where $\tH$ is the constant from Proposition~\ref{p:linH}.

In transition $4$, after a block of $(n-l)$ symbols, we transmit a block of $l$ symbols. This does not decrease mutual information.

In transition $5$, the output of the $(n-l)$-block is merged with the output of the $l$-block.
In order to un-merge the two outputs, one has to tell the receiver the length of the last output $l$-block.
This length is the sum of $l$ individual output lengths, corresponding to each input symbol; the entropy of individual length is upper-bounded by $\oh$ (see Proposition~\ref{p:finH}).
Thus, even if we tell the receiver all individual lengths (which are sufficient to compute the total length), we still reveal not greater than $l\cdot \oh$ bits.

Now, consider channel $V_{m|\tl|n}[\rho]:\mX^{m+l+n}\tosq \omY^2$.
We will need the auxiliary channel $W_{n,1}[\sigma\to\tau]:\mX^n\to\mS\times\omY^n$, where $\tau$ is considered as a part of the channel's output.
The channel can be represented as channel $W_{n,1}[\sigma]$, which additionally outputs the state after transmitting the last symbol:
\begin{align}
W_{n,1}[\sigma_1\toshort\tau](y_1^n|x_1^n)&=
\sum_{\sigma_2^n\in\mS^{n-1}} 
G_{\sigma_n,\tau} \cdot
\prod_{i=1}^{n-1} G_{\sigma_{i},\sigma_{i+1}}
\nonumber\\
&\times
\prod_{i=1}^n W_1[\sigma_i](y_i|x_i).
\label{eq:Wn1st}
\end{align}
Also, define the glued version of this channel as
\begin{align}
W_n[\sigma\toshort\tau](\oy|x_1^n)&=\sum_{y_1^n\in \bg_n^{-1}(\oy)}W_{n,1}[\sigma\toshort\tau](y_1^n|x_1^n).
\label{eq:Wnst}
\end{align}
The channel that does not output the final state and the channel that outputs the final state are related via marginalization over the final state:
\begin{align}
W_n[\sigma](\oy|x_1^n)&=\sum_\tau W_n[\sigma\toshort\tau](\oy|x_1^n),
\label{eq:Wnsigma}
\\
W_{n,1}[\sigma](y_1^n|x_1^n)&=\sum_\tau W_{n,1}[\sigma\toshort\tau](y_1^n|x_1^n).
\label{eq:Wn1sigma}
\end{align}
Now we can define the transition probabilities, corresponding to transmitting an $m$ and an $n$-block, separated by an $l$-block, which is then erased:
{\allowdisplaybreaks
\begin{align}
&V_{m|\tl|n}[\pi](y_1^2|x_1^3) \defeq
\nonumber\\&
\sum_{\substack{z_1\in\bg^{-1}_{m}(y_1)\\z_2\in\bg^{-1}_{n}(y_2)}} 
\sum_{\sigma_1^3\in\mS^3}\!\!
\pi_{\sigma_1}\! W_m[\sigma_1\!\toshort\!\sigma_2]\br{z_1|x_1}
G^l_{\sigma_2,\sigma_3} 
W_n[\sigma_3](z_2|x_3)
\nonumber \\ 
&\leq
\sum_{z_1^2,\sigma_1^3}~
\pi_{\sigma_1} W_m[\sigma_1\toshort\sigma_2]\br{z_1|x_1}
\pi_{\sigma_3} (1+\ve_l)
W_n[\sigma_3](z_2|x_3)
\nonumber \\
&\leq
(1+\ve_l)
\sum_{\sigma_1,z_1}
\pi_{\sigma_1} W_m[\sigma_1]\br{z_1|x_1}
\sum_{\sigma_3,z_2}
\pi_{\sigma_3}W_n[\sigma_3](z_2|x_3)
\nonumber \\&
=(1+\ve_l)V_m[\pi](y_1|x_1)V_n[\pi](y_2|x_3).
\label{eq:Vmln-l}
\end{align}
}

For brevity, we now introduce some shortcuts.
Let $U(y_1^2|x_1^2)=V_{m|\tl|n}[\pi](y_1^2|x_1,x',x_2)$.
Since $V_{m|\tl|n}[\pi](y_1^2|x_1^3)$ does not depend on $x_2$, $x'$ can be arbitrary.
Also, let $U'(y_1^2|x_1^2)=V_{m}[\pi](y_1|x_1)V_{n}[\pi](y_2|x_2)$.
Denote marginalized distribution $\tp(x_1,x_2)=\sum_{x'}p(x_1,x',x_2)$.
\begin{figure*}[b]
\vspace{-11pt}
\begin{align}
\hline\nonumber\\[-11pt]
&\bfI(V_{m|\tl|n}[\pi],p)=\bfI(U,\tp) 
=\sum_{x,y}\pWd{\tp}{U}(x,y) \brsq{
\log_2 U(y|x)
-\log_2 \pWyd{\tp}{U}(y)
}\nonumber \\
&\stackrel{6}{\leq}
\sum_{x,y}\set{(1-\ve_l)\pWd{\tp}{U'}(x,y) \log_2 \brsq{(1+\ve_l) U'(y|x)}
-(1+\ve_l)\pWd{\tp}{U'}(x,y) \log_2 \brsq{(1-\ve_l) \pWyd{\tp}{U'}(y)}}
\nonumber \\&
=(1-\ve_l)\bfE_{\pW{\tp}{U'}(x,y)}\brsq{\log_2 U'(y|x) + \log_2\br{1+\ve_l}}
-(1+\ve_l)\bfE_{\pW{\tp}{U'}(x,y)}
\brsq{\log_2 \pWyd{\tp}{U'}(y) + \log_2\br{1-\ve_l}}
\nonumber \\&
= \bfI(U',\tp)
-\ve_l \bfE_{\pW{\tp}{U'}(x,y)}\brsq{\log_2 U'(y|x)}
-\ve_l \bfE\brsq{\log_2 \pWyd{\tp}{U'}(y)}
+ \Gamma_l
= \bfI(U',\tp)
+ \ve_l \bfH_{\pW{\tp}{U'}(x,y)}\brsq{U'(y|x)}
+ \ve_l \bfH\brsq{\pWyd{\tp}{U'}(y)}
+ \Gamma_l
\nonumber \\&
\stackrel{7}{\leq} \bfI(V_m[\pi]\cdot V_n[\pi],\tp) + 2\ve_l(m+n)\tH+\Gamma_l.
\label{eq:UU'}
\end{align}
\end{figure*}%

Rewriting~\eqref{eq:Vmln-l} in shortcuts, one obtains $U(y|x)\leq (1+\ve_l)U'(y|x)$, and one can upper-bound the mutual information by \eqref{eq:UU'} (at the bottom of this page), where $$\Gamma_l=(1-\ve_l)\log_2(1+\ve_l)-(1+\ve_l)\log_2(1-\ve_l).$$%
The transition in $\stackrel{6}{\leq}$ is done as follows.
Each $\log$ in the first term adds a negative amount, so the multiplicative term before the first logarithm is lower bounded.
The logarithm itself is an increasing function, so its argument is upper-bounded.
Note that the upper bound might turn the value of the logarithm to positive, hence lower bounding of the multiplicative term decreases the value of the summand.
But since the original summand under $\sum$ was non-positive and became positive, it is still upper-bounded.
The second term is taken with a negative sign, so all the bounds are the opposite.

Bringing together \eqref{eq:VmnVmln} and \eqref{eq:UU'}, for $u_n=\bfI(V_n[\pi])$ we have
\begin{align}
u_{m+n}-(u_m+u_n) \leq l\tH+2\ve_l(m+n)\tH+\Gamma_l
\label{eq:umnres}
\end{align}
\begin{figure*}[t]
\vspace{-15pt}
\begin{align}
&\bfI(V_{\tl|n}[\rho],p)=\bfI(U,\tp) 
=
\sum_{x,y}\pWd{\tp}{U}(x,y) \brsq{
\log_2 U(y|x)-\log_2 \pWyd{\tp}{U}(y)
}
\nonumber \\&\geq
\sum_{x,y}\set{
(1+\ve_l)\pWd{\tp}{U'}(x,y) \log_2 \brsq{(1-\ve_l) U'(y|x)}
-(1-\ve_l)\pWd{\tp}{U'}(x,y) \log_2 \brsq{(1+\ve_l) \pWyd{\tp}{U'}(y)}
}
\nonumber \\&=
(1+\ve_l)\bfE_{\pW{\tp}{U'}(x,y)}\brsq{\log_2 U'(y|x) + \log_2\br{1-\ve_l}}
-(1-\ve_l)\bfE_{\pW{\tp}{U'}(x,y)}
\brsq{\log_2 \pWyd{\tp}{U'}(y) + \log_2\br{1+\ve_l}}
\nonumber \\&=
\bfI(U',\tp)
+\ve_l\set{\bfE_{\pW{\tp}{U'}(x,y)}\brsq{\log_2 U'(y|x)}+ \bfE\brsq{\log_2 \pWyd{\tp}{U}'(y)}}
+(1+\ve_l)\log_2\br{1-\ve_l}-(1-\ve_l)\log_2\br{1+\ve_l}
\nonumber \\&
=
\bfI(U',\tp)
- \ve_l \set{\bfH_{\pW{\tp}{U'}(x,y)}\brsq{U'(y|x)}+ \bfH\brsq{\pWyd{\tp}{U'}(y)}}
- \Gamma_l
\geq 
\bfI(U,\tp) - 2\ve_ln\tH-\Gamma_l
=
\bfI(V_n[\pi],\tp) - 2\ve_ln\tH-\Gamma_l.
\label{eq:ii}
\\
\hline \nonumber
\end{align}
\end{figure*}%
Let $N=m+n$, $l=\floor{\sqrt{N}}$ and denote by $f(N)$ the r.h.s. of \eqref{eq:umnres}.
Then, the first summand $l\tH \in O(\sqrt{N})$.
Recall \eqref{eq:ve-as} that $\ve_l$ is $O(e^{-Cl})$ for some constant $C$, thus, the second summand goes to zero, as well as the third one, $\Gamma_l\to 0$.
Thus, $f(N)\in O(\sqrt{N})$, so the series $\sum_{N} f(N)/N^2$ converges.
This implies the existence of $\lim_{N\to\infty}u_N/N=\bfI(V[\pi])$.

We can show that $\bfI(V[\rho])=\bfI(V[\pi])$ by the same trick of skipping $l$ symbols at the beginning\footnotemark.%
\footnotetext{
For transition from $V_n[\rho]$ to $V_n[\pi]$ we have to prove two transitions: $V_n[\rho] \to V_{\tl|n}[\rho]$ and $V_{\tl|n}[\rho] \to V_n[\pi]$.
The transition $V_n[\rho] \to V_{\tl|n}[\rho]$ is done as
\begin{align}
\bfI(V_n[\rho])&\leq \bfI(V_{l|n-l|l}[\rho]) \leq \bfI(V_{\tl|n-l|l}[\rho])+l\tH\leq \bfI(V_{\tl|n}[\rho])+l(\tH\!+\!\oh) \nonumber\\
\bfI(V_n[\rho])&\geq \bfI(V_{l|n-l|l}[\rho])-2l\oh \geq \bfI(V_{\tl|n}[\rho])-2l\oh.
\label{eq:fnI}
\end{align}
The difference is sublinear in $n$ ($l\sim \sqrt{n}$, $\oh$ and $\oH$ are constants), so the i-capacity of corresponding channel sequences is the same.
The transition $V_{\tl|n}[\rho] \to V_n[\pi]$ is done as follows.
Similarly to \eqref{eq:Vmln-l},
\begin{align}
(1-\ve_l)V_n[\pi](y|x_2) \leq V_{\tl|n}[\rho](y|x_1^2)\leq (1+\ve_l)V_n[\pi](y|x_2).
\label{eq:fnV}
\end{align}

Similarly to \eqref{eq:UU'},
\begin{align}
\bfI(V_{\tl|n}[\rho])\leq \bfI(V_n[\pi])+2\ve_l n \tH+\Gamma_l
\label{eq:fnIub}
\end{align}
and, using shortcuts $U(y|x_2)=\sum_{x_1}V_{\tl|n}[\rho](y|x_1^2)$, $U'(y|x)=V_n[\pi](y|x)$, we obtain \eqref{eq:ii}.
Since $2\ve_ln\tH \to 0$ and $\Gamma_l\to 0$ when $n\to\infty$, we can write
\begin{align}
\lim_{n\to\infty}\set{\bfI(V_{\tl|n}[\rho], p) - \bfI(V_{n}[\pi], \tp)}=0,
\label{eq:fnlim}
\end{align}
where $p$ is a distribution over $x_1^2\in \bF^l \times \bF^n$, and $\tp$ is a distribution over $x_2\in\bF^n$, which is the same as $p$, if the latter is marginalized over $x_1$.
However, by the definition of $V_{\tl|n}[\rho]$, the distribution over $x_1$ does not matter, since these input symbols are ignored by the channel.
So, for any $p$ there is $\tp$, and for any $\tp$ there is $p$, which satisfy \eqref{eq:ii}, so if we take maximum over $p$ and $\tp$, the limit is still zero even without dividing by $n$, so $\lim_{n\to\infty}\set{\bfI(V_{\tl|n}[\rho])-\bfI(V_n[\pi])}=0$.
}
To underline the fact that $\bfI(V[\rho])$ does not depend on $\rho$, we will denote the i-capacity by $\bfI(V)$.

\subsection{Proof of Theorem \ref{t:cviv}}

Recall that in this theorem, we consider only Markov-IDS channel sequences, for which the state channels only produce an output of length $\leq A$ for some constant $A$.
Also, throughout this section, we use a shortcut $V_n=V_n[\pi]$.


\textbf{Step 1. Transition $V\to\oV$.}
Consider channel sequence $\oV_{n}=V_{l_n|k_n|l_n|...|k_n}[\pi]:\mX^n\tosq\tmY^{2g_n}$, for which the receiver knows the outputs corresponding to all the input subvectors of lengths $l_n$ and $k_n$, where 
\begin{align}
k_n=\floor{\sqrt{n}},~l_n=\floor{\sqrt[4]{n}},~g_n=\bigg\lfloor\frac{n}{k_n+l_n}\bigg\rfloor
\label{eq:kl}
\end{align}
Note that with $n\to\infty$, the asymptotic of these variables are
\begin{align}
g_n, k_n \in O(\sqrt{n}),~l_n\in O(\sqrt[4]{n}).
\label{eq:allasy}
\end{align}
Below, we will omit index $_n$ and write $l$, $k$, $g$, meaning that they depend on $n$ by \eqref{eq:kl}.

Note that the channel $\oV_n$ gives to the receiver all the information that channel $V_n$ gives, plus the boundaries of $l$- and $k$-blocks.
For simplicity of notation, we assume that the last $r_n=n-g_n(k_n+l_n) \in O(\sqrt{n})$ input bits are ignored\footnote{Technically, the last $r_n$-block can be erased by an additional function, and the log-size of the inverse image of that function would be $\sim \log |\mY|^{Ar_n}\sim Ar_n$, which is $o(n)$. 
By Proposition~\ref{p:invimage}, such a function does not change the $\bfS$ property.}.

Thus, the channel $V_n$ can be represented as a pipeline of the channel $\oV_n$ and the merging function $\bg_{2g}$, which merges $2g$ vectors into one, i.e., $V_n=\bg_{2g}(\oV_n)$ (see \eqref{eq:fw} in the Appendix for the formal definition of a function of a channel/channel sequence).

Now, let us upper-bound the asymptotics of the log-size $\Phi_n$ of the inverse image of $\bg_{2g}$ in order to apply Proposition~\ref{p:invimage} to function sequence $\bg_{2g}$ and to finalize current step.
The total length of output of $\oV_n$ is at most $An$ (since we merged $n$ vectors).
The inverse image of $\bg_{2g}$ contains all possible partitions of vector of length at most $An$ into $2g$ subvectors.
It is equivalent to placing $2g-1$ commas between the $An$ symbols (with multiple commas in the same place allowed). The number of ways to place them can be (loosely) upper-bounded by $(An+1)^{2g-1}$, thus,
\begin{align}
\Phi_n&\leq\log_2\max_{y\in\tmY}\abs{\bg_{2g}^{-1}(y)}\leq \log_2(An+1)^{2g-1} 
\nonumber\\&
\in O(g\log_2 (An)) = O(\sqrt{n}\log n)\subset o(n).
\label{eq:Phinasy}
\end{align}
Thus, function sequence $\bg_{2g}$, applied to channel sequence $\oV_n$, satisfies Proposition~\ref{p:invimage}, and we obtain
\begin{align}
\bfI(V)=\bfI(\oV),~\bfS(V)\iff\bfS(\oV).
\label{eq:V-oV}
\end{align}

\textbf{Step 2. Transition $\oV\to\tV$.}
Now, denote by $\tV$ the sequence of channels $\tV_n=V_{\tl|k|\tl|k|\dots}[\pi]:\mX^n\to\tmY^g$.

The difference between channels $\tV_n$  and $\oV_n$ is that the $l$-blocks are deleted.
Let function $\bff_l$ delete all $l$-blocks.
Similarly to the previous step, we should prove that the logarithm $\Phi_n$ of the maximum inverse image of $\bff_l$ is in $o(n)$.
The function $\bff_l$ deletes $g$ blocks of output, each corresponds to $l$ channel input bits, and each channel input bit can result in maximum of $A$ channel output symbols (see the assumption at the beginning of this section), so the total number of deleted output symbols is $\leq Agl$.
Thus, the logarithm of the number of arguments of $\bff_l$, which result in the same input, can be upper-bounded as
\begin{align}
\Phi_n &=\max_{y\in \tmY^g}\abs{\bff_l^{-1}(y)}\leq \log_2\sum_{i=0}^{Agl}|\mY|^i \leq \log_2|\mY|^{Agl+1}
\nonumber\\
&= (Agl+1)\log_2|\mY|\in O(n^{3/4}) \subset o(n),
\end{align}
which means that the function sequence $\bff_{l}$ (recall that $l=l_n$ is a sequence) satisfies Proposition~\ref{p:invimage}, thus, $\bfI(\oV)=\bfI(\tV)$, $\bfS(\oV)\iff \bfS(\tV)$.
Together with \eqref{eq:V-oV}, we obtain
\begin{align}
\bfI(V)=\bfI(\tV),\; \bfS(V)\iff \bfS(\tV). 
\label{eq:V-tV}
\end{align}

\textbf{Step 3. Information stability of $\tV$.}
Similarly to \eqref{eq:Vmln-l}, one can obtain
\begin{align}
&(1-\ve_{l})^g \leq \frac{\tV_{n}(y|x)}{V^g_{k}(y|\mu(x))}
\leq (1+\ve_{l})^g
\label{eq:tVbounded}
\end{align}
where, for length-$n$ input $x$, function $\mu$ deletes the $l$-blocks and the last incomplete block, i.e. $\mu(x)=(x^{(1)},x^{(2)}, \dots, x^{(g)})$, and $x^{(i)}$ is the $i$-th $k$-block: $x^{(i)}=x_{i(k+l)-k+1}^{i(k+l)}$. 

Observe that the non-zero values of $V_k(y|x)$ satisfy
\begin{align}
&V_k(y|x)\geq \sum_{\substack{z\in\bg^{-1}_k(y)\\\sigma\in\mS^k}} \pi_{\sigma_1}\prod_{j=2}^k G_{\sigma_{j-1},\sigma_j} \prod_{i=1}^k W[\sigma_i](z_i|x_i)
\geq \psi^k,
\end{align}
where $\psi$ is the minimum\footnote{This is the only place where we actually need the assumption of finite-length output for a state channel. It allows us to use the minimum, which is never zero, instead of infimum, which might be zero.} non-zero state channel probability multiplied by the minimum non-zero transition probability of the Markov chain:
\begin{align}
\psi&=\min\br{\set{\pi_{\sigma_1},G_{\sigma_1,\sigma_2}}_{\sigma \in \mS^2}\setminus \set{0}}
\nonumber\\&
\times
\min\br{\set{W[\sigma](y|x)}_{\sigma\in\mS,y\in\tmY,x\in\mX}\setminus\set{0}}.
\end{align}
By Proposition~\ref{p:pvar}, there exists a sequence of distributions $\hp_k$, which asymptotically achieves the capacity of $V$, such that the variance of mutual information is
\begin{align}
\bfV(V_k,\hp_k) \in O(k^2).
\label{eq:varhat}
\end{align}

In what follows, we consider the input distribution $\tp$ for $\tV_n$, which is defined as follows.
All values in $l$-blocks and in the last $r$-block are filled with a pad --- say, with a fixed value $x_0\in\mX$.
Denote the set of such $x$ by $\mX^{(0)}_n$.
Then, the probability of a given $x$ is
\begin{align}
\tp_n(x_1^n)=
\begin{cases}
\prod_{i=1}^g \hp_k(x^{(i)}), & x\in\mX^{(0)}_n \\
0, & \text{otherwise}
\end{cases}
\label{eq:tp}
\end{align}

After defining the input distributions, we can consider the mutual information densities.
The fact that channel probabilities of $\tV_n$ and $V^g_k$ are close can be used to prove that mutual information densities of $\tV_n$ and $V^g_k$ are close too by using \eqref{eq:tVbounded} as
\begin{align}
&\bi(\tV_n,\tp_n,x,y)
=\log_2\frac{\tV_n(y|x)}{\pWyd{\tp_n}{\tV_n}(y)} 
\leq \log_2 \!\frac{(1+\ve_l)^g V^g_{k}(y|\mu(x))}{(1-\ve_l)^g \pWyd{\hp_k^g}{V^g_k}(y)}
\nonumber\\&
=
\bi(V^g_k,\hp_k^g,\mu(x),y)+g\log_2 \frac{1+\ve_l}{1-\ve_l}.
\label{eq:etvnvmq}
\end{align}
Using the same technique, one can obtain a similar lower bound:

\begin{align}
&\bi(\tV_n,\tp_n,x,y)
\geq \bi(V^g_k,\hp_k^g,\mu(x),y)-g\log_2 \frac{1+\ve_l}{1-\ve_l}.
\label{eq:etvnvmqlb}
\end{align}
We can rewrite \eqref{eq:etvnvmq}--\eqref{eq:etvnvmqlb} as
\begin{align}
\abs{\bi(\tV_n,\tp_n,x,y)\!-\!\bi(V^g_k,\hp_k^g,\mu(x),y)}\! \leq\! g\log_2 \frac{1+\ve_l}{1-\ve_l}.
\label{eq:itv-qiv-bounded}
\end{align}
As $\ve_l\to 0$, the logarithm in the r.h.s. of \eqref{eq:itv-qiv-bounded} behaves as $\ve_l$.
By \eqref{eq:ve-as} and \eqref{eq:kl}, $\ve_l\in O(e^{-D\sqrt[4]{n}})$ and $g\in O(\sqrt{n})$, so the limit of the r.h.s. is $0$.
Applying inequality $-|a-b|\leq a-b \leq |a-b|$ to \eqref{eq:itv-qiv-bounded}, we obtain
\begin{align}
\lim_{n\to\infty}
\underset{\pW{\tp_n}{\tV_n}}{\bfE}
\abs{
\bi(\tV_n,\tp_n,x,y)
-
\bi(V^g_k,\hp_k^g,\mu(x),y)
}
=0.
\label{eq:eitv-eq-eiqv}
\end{align}
Note that $\bfi(V^g_k, \hp_k^g, w, y)$ is equal to a sum of $g$ i.i.d. random variables.

Recall \eqref{eq:V-tV}, and write sufficient condition of information stability \eqref{eq:ise} as
{\allowdisplaybreaks
\begin{align}
&\lim_{n\to\infty}\underset{\pWd{\tp_n}{\tV_n}(x,y)}{\bfE}\abs{\frac{\bi(\tV_n,\tp_n,x,y)}{n}-\bfI(V)}
\nonumber\\&
\leq
\underbrace{
\lim_{n\to\infty}
\underset{\pWd{\tp_n}{\tV_n}(x,y)}{\bfE}
\abs{\frac{\bi(\tV_n,\tp_n,x,y)-\bi(V_k^g,\hp^g_k,\mu(x),y)}{n}}
}_{0~\text{by \eqref{eq:eitv-eq-eiqv}}}
\nonumber\\&+
\lim_{n\to\infty}
\underset{\pWd{\tp_n}{\tV_n}(x,y)}{\bfE}
\abs{\frac{\bi(V_k^g,\hp^g_k,\mu(x),y)-g\bfI(V_{k},\hp_k)}{n}}
\nonumber\\&+
\underbrace{\lim_{n\to\infty}
\abs{\frac{g\bfI(V_k,\hp_k)-g\bfI(V_k)}{n}}}_{0\ \text{since $\hp_k$ is capacity achieving}}
+
\underbrace{\lim_{n\to\infty}
\abs{\frac{g\bfI(V_{k})}{n}-\bfI(V)}}_{0\ \text{since $g/n\sim k$ and $\bfI(V_k)/k\to\bfI(V)$}}
\nonumber\\&
\overset{(7)}{\leq}
\lim_{n\to\infty} \underbrace{(1+\ve_l)^g}_{\to 1}\cdot
\underset{\pW{\hp_k^g}{V_k^g}(x,y)}{\bfE}
\abs{\frac{\bi(V_k^g,\hp^g_k,x,y)-g\bfI(V_k,\hp_k)}{n}}
\nonumber\\&
=
\lim_{n\to\infty} 
\frac{1}{k}\cdot \bfE
\abs{\frac{1}{g}\cdot \sum_{i=1}^g\bi(V_k,\hp_k,x_i,y_i)-\bfI(V_k,\hp_k)}
\nonumber\\&
\overset{(8)}{\leq}
\lim_{n\to\infty} 
\frac{1}{k}\cdot \sqrt{\bfE\brsq{
\br{\frac{1}{g}\cdot \sum_{i=1}^g\bi(V_k,\hp_k,x_i,y_i)-\bfI(V_k,\hp_k)}^2}}
\nonumber\\&
\overset{(9)}{=}
\lim_{k\to\infty}\frac{1}{k}\sqrt{\frac{\bfV(V_k,\hp_k)}{g}}
=
\lim_{k\to\infty}\sqrt{\frac{\bfV(V_k,\hp_k)}{k^3}}.
\label{eq:EitV}
\end{align}
}%
Note that during the transition $\overset{(7)}{\leq}$, we change the probability distribution for the expectation from $\pW{\tp_n}{\tV_n}(x,y)$ to $\pW{\hp_k^g}{V_k^g}(x,y)$. 
Since $\pWd{\tp_n}{\tV_n}(x,y)=\tp_n(x)\tV_n(y|x)$ and $\pW{\hp_k^g}{V_k^g}(x,y)=\hp^g_k(x)\tV^g_k(y|x)$,
by \eqref{eq:tVbounded} and \eqref{eq:tp}, this leads to a multiplicative term $(1+\ve_l)^g$.
The inequality $\overset{(8)}{\leq}$ is given by Jensen's inequality. 
The equality $\overset{(9)}{=}$ is true since the variance of an average of $g$ i.i.d. samples of a random variable is $g$ times smaller than the variance of the random variable.
By \eqref{eq:varhat} the last expression is zero.
Thus, the channel sequence $\tV$ satisfies Proposition~\ref{p:ise}, which implies $\bfS(\tV)$, which by \eqref{eq:V-tV} is equivalent to $\bfS(V)$.

The proof of $\bfS(V[\rho])$ for an arbitrary state distribution $\rho\neq \pi$ can be done easily by ignoring $l_n$ symbols at the beginning of transmission.
In other words, we insert a guard interval of size $l_n$.
Note that we already inserted infinitely many such guard intervals while transitioning from $V$ to $\tV$, and even this did not influence the information stability property (see \eqref{eq:V-oV}).

\section{Capacity Bounds for Channels with Markov Deletions}
\label{s:capa}

Having established the existence of the Shannon capacity of synchronization error channels with Markov memory, we analyze an example of such a channel in this section. The exemplary model is depicted in Fig.~\ref{fig:markovmodel}.
This is a Markov deletion channel model with two different states.
When the channel is in the first state, the deletion probability is low ($d$), while when it is in the second state, it is high ($D$).
Such a channel can be used as a model for a bursty deletion channel.

\begin{figure}
\centering
\begin{tikzpicture}[x=0.5cm,y=0.5cm]	
	\tikzstyle{block} = [circle, minimum size = 1cm, draw=black, thick, text=black];\
    \node[block] (s0)   at (0,  0) {$d$};
	\node[block] (s1)   at (4,  0) {$D$};
	\draw[-{Stealth},thick] (s0)--(s1) node[above, midway] {$\alpha$};
    \draw[-{Stealth}] (s1.north west) arc (30:150:1.5) node[above, midway]{$\beta$};
    \draw[-{Stealth},thick] (s0.north) arc (0:270:1) node[above=0.1, midway]{$1-\alpha$};
    \draw[-{Stealth},thick] (s1.north) arc (180:-90:1) node[above=0.1, midway]{$1-\beta$};
\end{tikzpicture}
\caption{Exemplary synchronization error channel with memory: deletion channel with two states.}
\label{fig:markovmodel}
\end{figure}

To obtain an upper bound on the capacity of a channel of this kind, we provide both the transmitter and the receiver with side information, employing the main ideas developed in \cite{fertonani10novel}. Namely, we divide the output into blocks corresponding to $k$ transmitted symbols. For each such block, we provide its boundaries, the initial state, and the state corresponding to the last transmitted symbol (even if the symbol is deleted). With this genie-aided information, we obtain a channel whose capacity is an upper bound on the capacity of the original deletion channel with memory. Note that the channel so-obtained is nothing but a sequence of DMCs, where each component DMC is known to transmitter and receiver. Thus, the capacity of the channel with side information is simply a linear combination of capacities of all $4\cdot (k+1)$ channels, which all are the $k$-bit Markov-IDS channels with fixed initial state, terminal state and the length of transmitted sequence (from $0$ to $k$). The transition probabilities for all conditioned channels can be computed straightforwardly, and the optimal input distribution, as well as the maximum mutual information for each of them, can be obtained by the Blahut-Arimoto algorithm. The weights in the linear combination are the probabilities of observing the specific DMCs, which can easily be computed using the Markov structure of the deletion process.

The average deletion probability of the two-state Markov-IDS channel is given by
\begin{align}
    \delta=\frac{\alpha D + \beta d}{\alpha+\beta}.
\label{eq:dmodel}
\end{align}
In the following, we will compare the capacity upper bounds for the two-state Markov-IDS channel with that of the i.i.d. deletion channel with identical average deletion probability $\delta$.

\begin{figure*}
\begin{minipage}{0.3\linewidth}
\centering
\includegraphics[width=\linewidth]{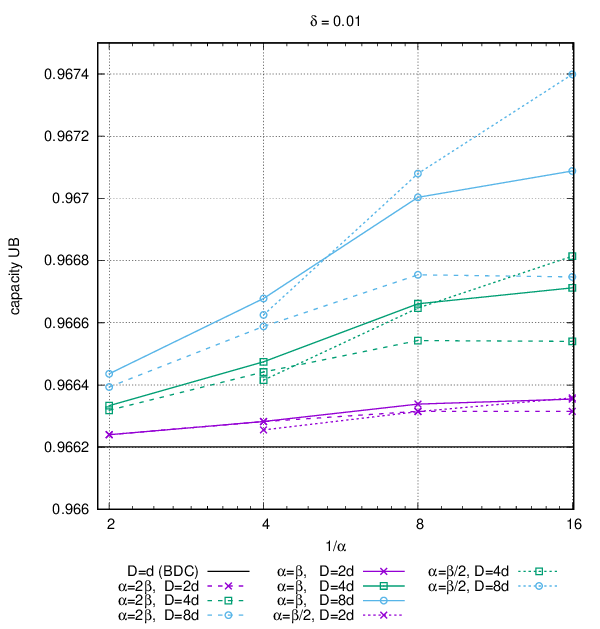}
\caption{The upper bound on capacity for $\delta=0.01$, BDC capacity $\leq 0.945$.}
\label{fig:d0.01}    
\end{minipage}\hfill
\begin{minipage}{0.3\textwidth}
\centering
\includegraphics[width=\linewidth]{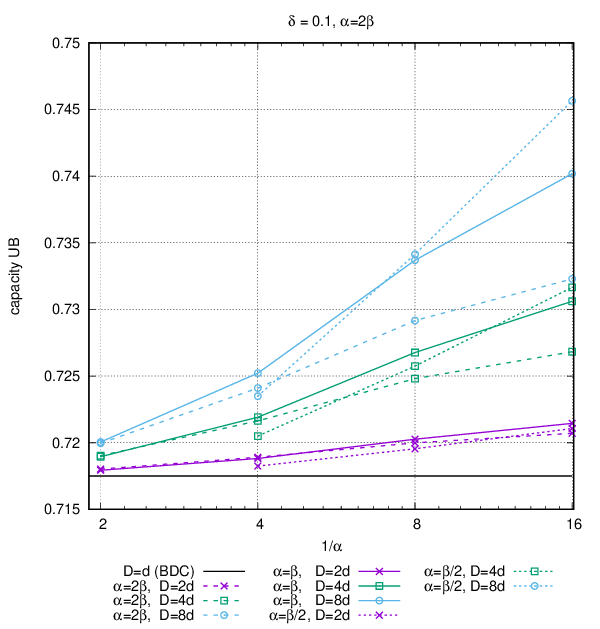}
\caption{The upper bound on capacity for $\delta=0.1$, BDC capacity $\leq 0.689$.}
\label{fig:d0.1}
\end{minipage}\hfill
\begin{minipage}{0.3\textwidth}
\centering
\includegraphics[width=\linewidth]{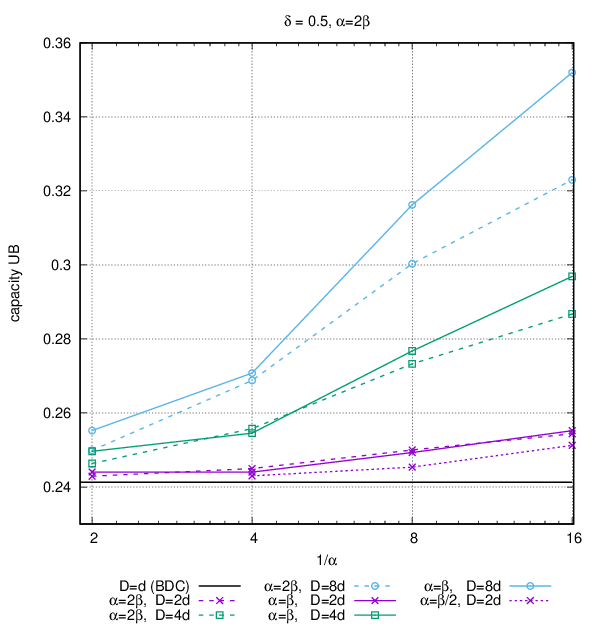}
\caption{The upper bound on capacity for $\delta=0.5$, BDC capacity $\leq 0.213$.}
\label{fig:d0.5}
\end{minipage}
\end{figure*}

In Fig.~\ref{fig:d0.01}--\ref{fig:d0.5}, the upper bounds on the capacity are provided for various Markov-IDS channels.
The parameters $\delta\in\set{0.01,0.1,0.5}$, ratio $D/d\in\set{2,4,8}$, value of $\alpha\in\set{\frac{1}{2},\frac{1}{4},\frac{1}{8},\frac{1}{16}}$ and ratio $\alpha/\beta\in\set{1/2,1,2}$ define the Markov-IDS channel uniquely.
The specific values of $d$ and $D$ can be computed by values of $D/d$, $\alpha$, $\beta$ and $\delta$ via \eqref{eq:dmodel}.
We assume that the initial state of the first block is chosen according to the stationary distribution $(\frac{\beta}{\alpha+\beta},\frac{\alpha}{\alpha+\beta})$.
The capacities were computed for $k=13$.

Each figure corresponds to a specific value of average deletion probability $\delta$. The point types and curve colors correspond to different values of the ratio $D/d$.
Intuitively, if the ratio $D/d$ is large, the channels in the states of the Markov chain are polarized, and the larger portion of deletions occur in the ``bad'' state, which makes the deletion process more bursty and more predictable, which results in the higher bound on capacity.
For comparison, the black curve is provided for the case of $d=D$, which corresponds to a binary i.i.d. deletion channel (with the same block-wise side information).
Although better bounds are available \cite{fertonani10novel} (we provide them in the caption to each figure), for comparison purposes, it is more suitable to compare bounds obtained in the same manner.

The dash types correspond to different values of $\alpha/\beta$.
The larger the value, the more probable the ``good'' state compared to the ``bad'' state. However, the average deletion probability is fixed, so the probabilities of states are balanced by values of $d$ and $D$ --- the more is $\alpha/\beta$, the larger become $d$ and $D$. One can see that the upper bound is not monotonous with respect to $\alpha/\beta$. 

Finally, the horizontal axis corresponds to the value of $1/\alpha$, which reflects the average number of channel uses spent in the ``good'' state. 
Since the relation $\alpha/\beta$ is fixed for each curve, when $\alpha$ is small ($1/\alpha$ is large), the switches between states in both directions are less probable, the sequence of states is more predictable, and the deletions are more probable to occur block-wise at the adjacent positions. Therefore, the upper bound on the capacity is larger for larger values of $1/\alpha$. 

In other words, if the memory is stronger, the deletions are more bursty, and hence, it becomes possible to transmit reliably using higher rates. 

\section{Summary and Conclusions}
\label{s:conc}
We proved that channels with synchronization errors for which a stationary and ergodic Markov chain governs the insertions/deletions, the information capacity of the corresponding channel sequence exists, and it is equal to the c-capacity, generalizing Dobrushin's classical result of Dobrushin on channels with i.i.d. synchronization errors. 
We generalized methods used in \cite{dobrushin59ageneral,dobrushin67shannons} as separate independent propositions to accomplish this goal. 
These propositions state that one can apply a function sequence to a channel sequence, changing neither the c-capacity nor the i-capacity of the latter if the function sequence has an asymptotically small preimage or asymptotically small probability of changing the output.
By applying such function sequences, one can bring the original channel sequence to another channel sequence, for which the capacity theorem or capacity bounds can be more easily derived, and one can be sure that the c-capacity~/~i-capacity is not changed under such transformation. 
The methodology may also be useful for deriving capacity theorems and capacity bounds for other non-trivial channel sequences.
We have also provided specific numerical examples, namely, capacity upper bounds with side information, and shown that having memory in the deletion process results in higher channel capacity for fixed average deletion probability.
Future research directions include further generalizations, for instance, including simultaneously the state's randomness and its dependence on input symbols, as observed in nanopore DNA sequencing.

\begin{appendices}
\section{Functions of Channels and Their Properties}
\label{s:funcapp}

In this appendix, we provide the proofs of propositions from Section~\ref{s:funcseq}.

\begin{proof}[Proof of Proposition~\ref{p:invimage} (a function with small preimage)]
First, note that 
\begin{align}
\underset{\pW{p_n}{V_n}}{\bfE}\brsq{\bi(V_n,p_n,x,y)}=
\underset{\pW{p_n}{U_n}}{\bfE}\brsq{\bi(V_n,p_n,x,f(y))}.    
\end{align}
Using Proposition~\ref{p:preimage} and the fact that $V_n$ is degraded with respect to $U_n$, for any sequence of input distributions $p$:
{\allowdisplaybreaks
\begin{align}
0&
\leq\lim_{n\to\infty}\frac{\bfI(U_n,p_n)-\bfI(V_n,p_n)}{n}
\nonumber\\&
=\lim_{n\to\infty}\underset{\pW{p_n}{U_n}}{\bfE}\brsq{\frac{\bi(U_n,p_n,x,y)-\bi(V_n,p_n,x,f_n(y))}{n}}
\nonumber\\&
\leq \lim_{n\to\infty}\underset{\pW{p_n}{U_n}}{\bfE}\brsq{\frac{\abs{\bi(U_n,p_n,x,y)-\bi(V_n,p_n,x,f_n(y))}}{n}} 
\nonumber\\&
\leq \lim_{n\to\infty}\frac{\Phi_n}{n}=0,
\label{eq:eabsdiff0}
\end{align}
}
which for optimal $p$ for $U$ implies that $\bfI(U)=\bfI(V)$, if either $\bfI(U)$ or $\bfI(V)$ exists.
Equation \eqref{eq:eabsdiff0} implies that for any $p_n$,
\begin{align}
\lim_{n\to\infty}\underset{\pW{p_n}{U_n}}{\bfE}\brsq{\frac{\abs{\bi(U_n,p_n,x,y)\!-\!\bi(V_n,p_n,x,f_n(y))}}{n}}\!=\!0.
\end{align}
Using the same technique as in Proposition~\ref{p:ise}, it can be easily shown that for any $\delta>0$ and $p_n\in\bD_{\mX^n}$,
\begin{align}
\lim_{n\to\infty}\Pr_{\pW{p_n}{U_n}}\brsq{\frac{\abs{\bi(U_n,p_n,x,y)-\bi(V_n,p_n,x,f_n(y))}}{n}\!>\!\delta}=0.
\label{eq:pr-abs-diff-0}
\end{align}
If $\bfS(U)$, then for arbitrary $\delta>0$, both \eqref{eq:is} (but for $U$ instead of $V$) and \eqref{eq:pr-abs-diff-0} can be satisfied for $\delta/2$ simultaneously using some sequence of distributions $p^*$.
The equality $\bfI(U)=\bfI(V)$, proven above, and the triangle inequality imply
\begin{align}
&\lim_{n\to\infty}\Pr_{\pW{p^*_n}{U_n}}\brsq{\abs{\frac{\bi(V_n,p^*_n,x,f_n(y))}{n}-\bfI(V)}>\delta}
\nonumber\\&
\leq\lim_{n\to\infty}\Pr_{\pW{p^*_n}{U_n}}\brsq{\frac{\abs{\bi(U_n,p^*_n,x,y)-\bi(V_n,p^*_n,x,f_n(y))}}{n}>\frac{\delta}{2}}
\nonumber\\&
+\lim_{n\to\infty}\Pr_{\pW{p^*_n}{U_n}}\brsq{\abs{\frac{\bi(U_n,p^*_n,x,y)}{n}-\bfI(U)}>\frac{\delta}{2}}=0.
\end{align}
This implies $\bfS(V)$.
The implication $\bfS(V)\implies\bfS(U)$ can be shown similarly.

To prove that $\bfC(U)=\bfC(V)$, we use the expression \eqref{eq:cformula} for  c-capacity.
Assume that $\bfC(V)<\bfC(U)=\alpha$.
Then,
\begin{align}
&\forall \ve>0,\delta>0, \exists n^*_0:\forall n \geq n^*_0, \exists p^*_n:
\nonumber\\&
\Pr_{\pW{p^*_n}{U_n}(x,y)}\brsq{\frac{\bi(U_n,p^*,x,y)}{n}> \alpha-\ve}>1-\delta. 
\label{eq:CUalpha}\\
&\exists \ve^*>0,\delta^*>0: \forall n_0, \exists n^* \geq n_0: \forall p_{n^*},
\nonumber\\&
\Pr_{\pW{p_{n^*}}{U_{n^*}}(x,y)}\brsq{\frac{\bi(V_{n^*},p,x,f_{n^*}(y))}{n^*}\leq \alpha-\ve^*}\geq\delta^*.
\label{eq:CVnotalpha}
\end{align}
Since \eqref{eq:CUalpha} holds for any $\ve>0$, $\delta>0$, we can let $\ve=\ve^*/2$, $\delta=\delta^*/2$, where $\ve^*$ and $\delta^*$ are from \eqref{eq:CVnotalpha}, and obtain \eqref{eq:epsdeltasim} (see the top of the next page).

\begin{figure*}[t]
\begin{align}
\exists \ve^*>0,\delta^*>0: 
\begin{cases}
\exists n^*_0:\forall n \geq n^*_0, \exists p_n^*:~
\underset{\pW{p_n^*}{U_n}(x,y)}{\Pr}
\brsq{\bi(U_n,p_n^*,x,y)/n> \alpha-\ve^*/2}>1-\delta^*/2
\\
\forall n_0, \exists n^* \geq n_0: \forall p_{n^*},
~\underset{\pW{p_{n^*}}{U_{n^*}}(x,y)}{\Pr}\brsq{\bi(V_{n^*},p_{n^*},x,f_{n^*}(y))/{n^*}\leq \alpha-\ve^*}\geq\delta^*
\end{cases}
\label{eq:epsdeltasim}
\\ \hline \nonumber 
\end{align}    
\end{figure*}

So, the first inequality holds for all $n$ starting from $n^*_0$, and for each $n$ specific $p_n^*$ should be picked up.
The second inequality holds for some $n$ larger than arbitrary $n_0$, for all $p_n$.
Bringing these two conditions together, we can satisfy both inequalities as follows:
\begin{itemize}
\item for any $n_0$, pick up the value of $n^*\geq \max\set{n_0,n^*_0}$, for which the second inequality holds;
\item since $n^*\geq n^*_0$, the first equality still holds for some specific selection of $p^*_{n^*}$. The second inequality will also hold for $p_{n^*}^*$, since it holds for any $p_{n^*}$.	
\end{itemize}
The above procedure results in
\begin{align}
&\exists \ve^*>0,\delta^*>0: \forall n_0, \exists n\geq n_0, \exists p_n:
\nonumber\\&
\begin{cases}
\displaystyle \Pr_{\pW{p_n}{U_n}(x,y)}\brsq{\frac{\bi(U_n,p_n,x,y)}{n}> \alpha-\frac{\ve^*}{2}}>1-\frac{\delta^*}{2}
\\
\displaystyle \Pr_{\pW{p_n}{U_n}(x,y)}\brsq{\frac{\bi(V_{n},p_{n},x,f_{n}(y))}{n}\leq \alpha-\ve^*}\geq\delta^*
\end{cases}
\end{align}
Both probabilities are defined over the same probability space.
The sum of the two probabilities is greater than $1+\delta^*/2$, thus,
\begin{align*}
\Pr_{\pW{p_n}{U_n}(x,y)}\set{\frac{\bi\br{U_n,p_n,x,y}-\bi\br{V_{n},p_{n},x,f_{n}(y)}}{n}\!>\!\frac{\ve^*}{2}}\!>\!\frac{\delta^*}{2},
\end{align*}
which contradicts \eqref{eq:pr-abs-diff-0}.
\end{proof}

\begin{proof}[Proof of Proposition~\ref{p:prb} (a semi-bijective function)]
Let $U':\mX\tosq\mB$ be a channel equivalent to $U$ but with the output restricted to $\mB$:
\begin{align}
U'(y|x)={U(y|x)}/{\beta(x)}.
\end{align}

By definition, $\bfI(U,p)=\bfE_{\pW{p}{U}(x,y)}\brsq{\bi(U,p,x,y)}$.
The expectation can be expressed by two separate summations over $\mA$ and $\mB$:
\begin{align}
\underset{\pW{p}{U}(x,y)}{\bfE}\brsq{\bi(U,p,x,y)}
&=\underbrace{\sum_{y\in \mA,x}p(x) U(y|x)\bi(U,p,x,y)}_{a}
\nonumber\\&
+\underbrace{\sum_{y\in \mB,x}p(x) U(y|x)\bi(U,p,x,y)}_{b}
\end{align}
Let us analyze the term $b$.
First, introduce the modified distribution $\tp(x)=p(x)\beta(x)/\Delta$, where $\Delta=\sum_x p(x)\beta(x)$.
Note that $0\leq \Delta \leq \obeta$.
Substituting $p(x)U(y|x)=\Delta\tp(x)U'(y|x)$ for $y\in\mB$, one obtains
\begin{align*}
&b
=\sum_{y\in\mB,x\in\mX}p(x)U(y|x)\log_2\frac{U(y|x)}{\sum_w p(w)U(y|w)}
\\&
=\Delta\cdot\sum_{y\in\mB,x\in\mX}\tp(x)U'(y|x)
\log_2\frac{\beta(x)U'(y|x)}{\Delta\sum_w \tp(w)U'(y|w)}
\\&
=\Delta\cdot\bfI(U',\tp)+\Delta\sum_{y\in\mB,x\in\mX}\tp(x)U'(y|x)\log_2\frac{\beta(x)}{\Delta}
\\&
= \Delta\br{\bfI(U',\tp)-\log_2 \Delta}+\!\!\!\!\sum_{y\in\mB,x\in\mX}\!\!p(x)U'(y|x)\beta(x)\log_2 \beta(x)
\end{align*}
Using inequality $-\frac{1}{e\ln 2}\leq x\log_2 x \leq 0$ when $0<x\leq 1$, one obtains
\begin{align}
-\Delta\log_2\Delta-\frac{1}{e\ln 2}\leq b \leq \Delta \log_2\frac{|\mX|}{\Delta}
\label{eq:b}
\end{align}
The same can be done for the channel $V$:
\begin{align}
&\bfE_{\pW{p}{V}(x,y)}\brsq{\bi(V,p,x,y)}
=\sum_{y\in \mA'}\sum_{x\in\mX}p(x) V(y|x)\bi(V,p,x,y)
\nonumber\\&
+\underbrace{\sum_{y\in \mB'}\sum_{x\in\mX}p(x) V(y|x)\bi(V,p,x,y)}_{b'}
\nonumber\\&
=\sum_{y\in \mA}\sum_{x\in\mX}p(x) U(y|x)\bi(U,p,x,y)+b'=a+b'
\end{align}
Note that $\beta(x)=\sum_{y\in\mB}U(y|x)=\sum_{z\in\mB'}V(z|x)$, so, using the auxiliary channel $V'(z|x):\mX\tosq\mB'$ defined as $V'(z|x)=V(z|x)/\beta(x)$, one obtains for any $z\in\mB'$, $p(x)V(z|x)=\Delta \tp(x)V'(z|x)$.
Similarly to $b$, the value of $b'$ can be represented as
\begin{align*}
&b'=\sum_{z\in \mB',x\in\mX}p(x) V(z|x)\bi(V,p,x,z)
\\&=\Delta\cdot\!\!\sum_{z\in\mB',x\in\mX}\tp(x)V'(z|x)\log_2\frac{\beta(x)V'(z|x)}{\Delta\sum_w \tp(w)V'(z|w)}
\\&
= \Delta\br{\bfI(V',\tp)-\log_2 \Delta}+\!\!\!\!\!\sum_{z\in\mB',x\in\mX}\!\!\!\!\!p(x)V'(y|x)\beta(x)\log_2 \beta(x)
\end{align*}
so \eqref{eq:b} holds for $b'$ as well.
This implies that for all $p\in\bD_\mX$,
\begin{align}
&\bfI(U,p)-\bfI(V,p) = b-b' \leq 
\Delta \log_2\frac{|\mX|}{\Delta}
\nonumber\\&+\Delta\log_2\Delta+\frac{1}{e\ln 2}
=\Delta \log_2|\mX|+\frac{1}{e\ln 2}.
\end{align}
Since $\Delta\leq \obeta$, \eqref{eq:uvbeta} holds.
\end{proof}

\begin{proof}[Proof of Proposition~\ref{p:pr0} (a function that does nothing almost surely)]
If either $\bfI(U)$ or $\bfI(V)$ exists, then, using Proposition~\ref{p:prb} and distributions $p^*_n$, optimal for $U_n$,
\begin{align*}
0&\leq\bfI(U)-\bfI(V)\leq\lim_{n\to\infty}\frac{\bfI(U_n,p^*_n)-\bfI(V_n,p^*_n)}{n}
\\&
\leq\lim_{n\to\infty}\br{\frac{\obeta_nn\log_2|\mX|}{n}+\frac{1}{ne\ln 2}}=0,
\end{align*}
so $\bfI(U)=\bfI(V)$.

Assume that $\bfS(U)$ is true.
Note that, for any $p_n$,
\begin{align}
&\lim_{n\to\infty}\Probd{\pW{p_n}{U_n}(x,y)}{\bfi(U_n,p_n,x,y)\neq \bfi(V_n,p_n,x,f_n(y))}
\nonumber\\&=    \lim_{n\to\infty}\Probd{\pW{p_n}{U_n}(x,y)}{y\in\mB}
\leq \lim_{n\to\infty}\obeta_n=0.
\label{eq:ineqi}
\end{align}

Since $\bfI(U)=\bfI(V)$, there exists a sequence of distributions $(p^*_n)_{n\in\bN}$, such that for any $\delta>0$:
{\allowdisplaybreaks
\begin{align}
&\lim_{n\to\infty}\Probd{\pW{p^*_n}{V_n}(x,z)}{\abs{\frac{\bfi(V_n,p^*_n,x,z)}{n}-\bfI(V)}>\delta}
\nonumber\\&
=\lim_{n\to\infty}\Probd{\pW{p^*_n}{U_n}(x,y)}{\abs{\frac{\bfi(V_n,p^*_n,x,f_n(y))}{n}-\bfI(U)}>\delta}
\nonumber\\&
\leq \lim_{n\to\infty}\bigg[\Probd{\pW{p^*_n}{U_n}(x,y)}{\bfi(V_n,p^*_n,x,f_n(y))\neq \bfi(U_n,p^*_n,x,y)}
\nonumber\\&
+\Probd{\pW{p^*_n}{U_n}(x,y)}{\abs{\frac{\bfi(U_n,p^*_n,x,y)}{n}-\bfI(U)}>\delta}\bigg],
\end{align}
}%
and by \eqref{eq:ineqi} and information stability of $U$, the limit is $0$.
Thus, $\bfS(U)\implies \bfS(V)$.
The other implication can be shown similarly.

According to \eqref{eq:cformula}, $\bfC(U)=\bfC(V)$ can be proved by analyzing the probability $\Pr\set{\bi(V_n,p_n,x,y)\leq\alpha n}$ for some value of $\alpha$:
\begin{align}
&\Pr_{\pW{p_n}{V_n}(x,y)}\brsq{\bi(V_n,p_n,x,y)\leq \alpha n}
\nonumber\\&=
\Pr_{\pW{p_n}{U_n}(x,y)}\brsq{\bi(V_n,p_n,x,f(y))\leq \alpha n}
\nonumber\\&
\leq \Pr_{\pW{p_n}{U_n}(x,y)}\brsq{\bi(V_n,p_n,x,f(y))\neq\bi(U_n,p_n,x,y)}
\nonumber\\&
+\Pr_{\pW{p_n}{U_n}(x,y)}\brsq{\bi(U_n,p_n,x,y)\leq\alpha n}
\label{eq:pruv}
\end{align}
The first probability goes to zero according to \eqref{eq:ineqi}. Thus, 
\begin{align}
&\lim_{n\to\infty}\min_{p_n}\Pr_{\pW{p_n}{U_n}(x,y)}\brsq{\bi(U_n,p_n,x,y)\leq \alpha n}=0 
\nonumber\\&\implies
\lim_{n\to\infty}\min_{p_n}\Pr_{\pW{p_n}{V_n}(x,y)}\brsq{\bi(V_n,p_n,x,y)\leq \alpha n}=0,
\end{align}
which means that the set of such $\alpha$'s for $U_n$ is a subset of such set for $V_n$, which implies $\bfC(V)=\bfC(U)$.
\end{proof}

\section{Perturbation of an optimal distribution}
\label{s:pert}

In this section, we investigate how a perturbation of the optimal distribution influences the i-capacity.
We also provide the order of growth for the mutual information variance.

\begin{proposition}[Optimal distribution perturbation]
\label{p:peps}
Consider a channel sequence $V_n:\mX^n\tosq\mY_n$ with optimal input distributions $\op_n$.
Denote by $\vp_n \geq 0$ the logarithm of the infimum of non-zero channel probabilities:
\begin{align}
\vp_n = -\log_2\inf\br{\set{V_n(y|x)}_{x,y}\setminus\set{0}}.
\label{eq:phin}
\end{align}
For any sequence $\delta$ with $0\leq \delta_n \leq 1$, consider the sequence of perturbed distributions $\hp^\infty$, which are mixtures of optimal distributions and uniform distributions over $\mX^n$:
\begin{align}
\hp_n(x)=\br{1-\delta_n}\op_n(x)+\frac{\delta_n}{|\mX|^n}.
\label{eq:ppert}
\end{align}
Then, the asymptotic difference of the mutual information capacities is
\begin{align}
\bfI(V_n) - \bfI(V_n,\hp_n) \in O(\delta_n(n+\vp_n)).
\label{eq:pertopt}
\end{align}

In particular\footnote{The intersection is needed for the cases when $\vp_n\to 0$, i.e., for channel sequences, where the minimum non-zero transition probability goes to 1.}, if $\delta_n \in o(1/\vp_n)\cap o(1)$, the distributions $\hp_n$ asymptotically achieve the i-capacity of $V$ (if such exists):
\begin{align}
\lim_{n\to\infty}\frac{\bfI(V_n,\hp_n)}{n}=\bfI(V).
\label{eq:pertach}
\end{align}
Moreover, if we choose $\delta_n \in o(1/\vp_n) \cap o\br{1/n}$, then $\hp_n$ also achieve the i-capacity for each channel $V_n$:
\begin{align}
\lim_{n\to\infty}\brsq{\bfI(V_n)-\bfI(V_n,\hp_n)}=0.
\label{eq:pertstr}
\end{align}
\end{proposition}
\begin{proof}
The equations \eqref{eq:pertach} and \eqref{eq:pertstr} directly follow from \eqref{eq:pertopt}.
Observe that, by definition, $\bfI(V_n,\hp_n)\leq \bfI(V_n)$.
Thus, it is sufficient to show that
\begin{align}
\bfI(V_n,\hp_n) \geq \bfI(V_n)+O(\delta_n(n+\vp_n)).
\label{eq:pertosm}
\end{align}%
\begin{figure*}[t]
\begin{align}
&\bfI(V_n,\hp_n)
=~
\sum_{x}\br{\op_n(x)\br{1-\delta_n}+\frac{\delta_n}{\xi^n}}
\sum_y V_n(y|x) \log_2\frac{V_n(y|x)}{\pWyd{\hp_n}{V_n}(y)}
=
\br{1-\delta_n}\underbrace{\sum_{x}\op_n(x)\sum_y V_n(y|x)\log_2\frac{V_n(y|x)}{\pWyd{\op_n}{V_n}(y)}}_{\bfI(V_n)}
\nonumber \\
+&
\br{1-\delta_n}
\underbrace{
\sum_{x}\op_n(x)\sum_y V_n(y|x)\log_2\frac{\pWyd{\op_n}{V_n}(y)}{\pWyd{\hp_n}{V_n}(y)}
}_{D_\text{KL}\br{\pWy{\op_n}{V_n} || \pWy{\hp_n}{V_n}}\geq 0}
+
\frac{\delta_n}{\xi^n}
\underbrace{
\sum_{x}\sum_{y}V_n(y|x)\underbrace{\log_2 V_n(y|x)}_{\geq-\vp_n}
}_{\geq -\xi^n \vp_n}
+
\underbrace{\frac{\delta_n}{\xi^n}\sum_{x}\sum_{y}V_n(y|x)
\log_2 \frac{1}{\pWyd{\hp_n}{V_n}(y)}}_{\geq 0}
\nonumber \\
\geq &~
\br{1-\delta_n}\bfI(V_n)-\delta_n\vp_n
\geq 
\bfI(V_n)+O(\delta_n(n +\vp_n))
\label{eq:lbpert}
\\
\hline \nonumber  
\end{align}
\end{figure*}%
With $\xi=|\mX|$, it is shown in \eqref{eq:lbpert} at the top of the next page.
\end{proof}

The perturbed distributions $\hp_n$ have a nice property: they are non-zero everywhere, namely, $\hp_n(x)\geq\delta_n/|\mX|^n$.
We use it to obtain the following result.
\begin{proposition}[The order of mutual information variance]
\label{p:pvar}
Consider a channel sequence $V_n:\mX^n\tosq\mY_n$ with finite $|\mX|$.
Assume that the limit $\bfI(V)$ exists and $\displaystyle\lim_{n\to\infty}\vp_n \neq 0$.
Then, there exists a capacity-achieving sequence of input distributions $\hp$, such that
\begin{align}
\bfV(V_n,\hp_n)\in O(\vp_n^2+n^2).
\label{eq:pertvar}
\end{align}
\end{proposition}
\begin{proof}
For any optimal distribution $p_n$, consider a perturbed distribution $\hp_n$, defined in \eqref{eq:ppert}.
Then,
\begin{align}
&\bfV(V_n,\hp_n)
\leq
\underset{\pW{\hp_n}{V_n}(x,y)}{\bfE}\brsq{\bfi^2(V_n,\hp_n,x,y)}
\nonumber \\&
\leq
\sum_{x}\hp_n(x) \sum_y V_n(y|x)\brsq{\underbrace{\log_2^2 V_n(y|x)}_{\leq \vp^2_n}+\log_2^2 \pWyd{\hp_n}{V_n}(y)}
\nonumber \\&
\leq 
\vp^2_n + \log_2^2\frac{\delta_n}{\xi^n 2^{\vp_n}} \in O(\vp_n^2+\log_2^2 \delta_n+n^2).
\label{eq:var}
\end{align}
Let $\delta_n=1/\vp_n \in o(n/\vp_n)$. 
Then, by Proposition~\ref{p:peps}, $\hp$ is capacity-achieving.
Furthermore, $\log_2^2\delta_n \in O(\log^2\vp_n) \subseteq O(\vp_n^2)$, so the term $\log_2^2\delta_n$ in \eqref{eq:var} can be ignored.
\end{proof}

\end{appendices}

\bibliographystyle{IEEEtran}
\bibliography{coding}

\end{document}